\pdfoutput=1

\documentclass[11pt]{article}

\usepackage[preprint]{acl}
\usepackage{times}
\usepackage{latexsym}  
\usepackage[T1]{fontenc}
\usepackage[utf8]{inputenc}
\usepackage{microtype}
\usepackage{inconsolata}
\usepackage{graphicx}
\usepackage{bm}
\usepackage{amsmath}
\usepackage{amsthm}
\usepackage{booktabs}
\usepackage{multirow}
\usepackage{bigstrut}
\usepackage{hyperref}
\usepackage{graphicx}
\usepackage{subcaption}
\usepackage{tcolorbox}
\usepackage{tabularx}
\usepackage{titlesec}
\usepackage{epigraph}

\theoremstyle{theorem}
\newtheorem{mytheory}{Theorem}
\theoremstyle{proof}

\titlespacing*{\section}{0pt}{4pt}{4pt}

\titlespacing*{\subsection}{0pt}{4pt}{3pt}

\titlespacing*{\subsubsection}{0pt}{4pt}{3pt}

\title{TrendSim: Simulating Trending Topics in Social Media Under Poisoning Attacks with LLM-based Multi-agent System}

\author{
  \textbf{Zeyu Zhang\textsuperscript{1}},
  \textbf{Jianxun Lian\textsuperscript{2}},
  \textbf{Chen Ma\textsuperscript{1}},
  \textbf{Yaning Qu\textsuperscript{1}},
  \textbf{Ye Luo\textsuperscript{1}},
  \textbf{Lei Wang\textsuperscript{1}},
  \\
  \textbf{Rui Li\textsuperscript{1}},
  \textbf{Xu Chen\textsuperscript{1}},
  \textbf{Yankai Lin\textsuperscript{1}},
  \textbf{Le Wu\textsuperscript{3}},
  \textbf{Xing Xie\textsuperscript{2}},
  \textbf{Ji-Rong Wen\textsuperscript{1}},
\\
\\
  \textsuperscript{1}Renmin University of China,
  \textsuperscript{2}Microsoft Research Asia,
\\
  \textsuperscript{3}Hefei University of Technology
\\
  \texttt{\{zeyuzhang,xu.chen\}@ruc.edu.cn}
}

\begin{document}
\setlength{\abovedisplayskip}{3pt}
\setlength{\belowdisplayskip}{3pt}
	
\maketitle
\begin{abstract}
	
Trending topics have become a significant part of modern social media, attracting users to participate in discussions of breaking events.
However, they also bring in a new channel for poisoning attacks, resulting in negative impacts on society.
Therefore, it is urgent to study this critical problem and develop effective strategies for defense.
In this paper, we propose TrendSim, an LLM-based multi-agent system to simulate trending topics in social media under poisoning attacks.
Specifically, we create a simulation environment for trending topics that incorporates a time-aware interaction mechanism, centralized message dissemination, and an interactive system.
Moreover, we develop LLM-based human-like agents to simulate users in social media, and propose prototype-based attackers to replicate poisoning attacks.
Besides, we evaluate TrendSim from multiple aspects to validate its effectiveness.
Based on TrendSim, we conduct simulation experiments to study four critical problems about poisoning attacks on trending topics for social benefit.

\end{abstract}

\section{Introduction}
	\label{sec:introduction}
	Trending topics have become a significant part of modern social media platforms in recent years, which refer to the topics that draw public attention and widespread discussions within a short period, such as that in \textit{Weibo Hot Searches}\footnote{\url{https://s.weibo.com/top/summary}} and \textit{Twitter Trends}\footnote{\url{https://twitter.com/explore/tabs/trending}}.
	Each trending topic typically includes a headline, a description, and numerous comments.
	Compared with conventional social media contents, trending topics always emerge explosively, and are displayed in a highlight section to ensure users' engagement in current discussions.
	However, these factors also amplify the influence of poisoning attacks on users in social media platforms.
	Such attacks often manipulate or distort information in the comments on trending topics to mislead users and spread misinformation.
	They can result in numerous negative impacts, such as misguiding facts, provoking conflicts, and even destroying social trusts, which are harmful to our society. However, this critical problem remains inadequately studied yet.
	
	Recently, large language models (LLMs) have exhibited human-like capabilities~\cite{zhao2023survey,wang2023survey}, and several studies have proposed to design LLM-based human-like agents to conduct social simulations~\cite{gao2023large,gao2023s,kovavc2023socialai}. By analyzing their results, researchers can draw deep insights into human behaviors, and devise effective policies for social benefit~\cite{hua2023war}.
	However, previous simulation frameworks have limitations that make them unsuitable for simulating trending topics.
	First of all, most frameworks employ round-based interactions without taking time into consideration~\cite{wang2023large}, despite the fact that trending topics are highly time-sensitive. Second, most social simulations are designed for peer-to-peer interactions where agents primarily communicate with their neighbors~\cite{gao2023s}. However, due to the individual section of trending topics, their message dissemination should be centralized like a hub. Moreover, previous methods seldom focus on dynamic psychological conditions during simulations, and overlook the problem of poisoning attacks in trending topics.

	To address these limitations, in this paper, we propose an LLM-based multi-agent system, named \textbf{TrendSim}, to simulate trending topics in social media under poisoning attacks.
	Specifically, our framework designs a time-aware interaction mechanism and centralized message dissemination to adapt to the trending topic scenario, and implement details of a multi-agent interactive system. Besides, we design LLM-based human-like agents with a perception, a memory, and an action module, in order to simulate user behaviors and reflect their psychological conditions. Moreover, we create prototype-based attackers with different targets to generate poisoning attacks.
	We conduct extensive evaluations to verify the effectiveness of our simulation framework from multiple aspects.
	Based on TrendSim, we study four critical problems of the poisoning attacks on trending topics in social media, analyzing the results and providing suggestions for defense.
	Our work is the first one that focuses on the poisoning attack problem in trending topics with LLM-based social simulations.
	
	However, as an initial study in this new area, we should emphasize some facts in our work.
	First, because the implementation of trending topics varies across social media platforms, our work abstracts a common and reasonable implementation, based on certain assumptions.
	Second, the existence of numerous invisible factors in the real world makes it impractical to simulate all details with complete accuracy. Therefore, our research aims to provide an interpretable simulation process and deliver evolutionary conclusions under reasonable assumptions, rather than replicating every detail of reality.
	
	Our contributions are summarized as follows:

	\noindent $\bullet$ We propose an LLM-based multi-agent system, named TrendSim, to simulate trending topics in social media under poisoning attacks. We design the time-aware interaction mechanism, centralized message dissemination, and interactive multi-agent system to model trending topics in social media.
	
	\noindent $\bullet$ We develop LLM-based human-like agents with a perception, a memory, and an action module to simulate users in social media platforms. We create prototype-based attackers to generate various poisoning attacks in our simulation.
	
	\noindent $\bullet$ We conduct extensive evaluations of our simulation framework. Based on TrendSim, we study four critical problems of poisoning attacks on trending topics in social media.

	The rest of our paper is organized as follows. First, we present related works in Section~\ref{sec:related_work}. Then, we demonstrate the details of TrendSim in Section~\ref{sec:framework}, and conduct evaluations in Section~\ref{sec:alignment}.
	After that, we conduct simulation experiments based on TrendSim in Section~\ref{sec:experiment}.
	Finally, we draw conclusions in Section~\ref{sec:conclusion}, and further discuss limitations and ethical impacts at the end of this paper.
	
	\begin{figure*}[t]
		\centering
		\setlength{\fboxrule}{0.pt}
		\setlength{\fboxsep}{0.pt}
		\fbox{
			\includegraphics[width=1.0\linewidth]{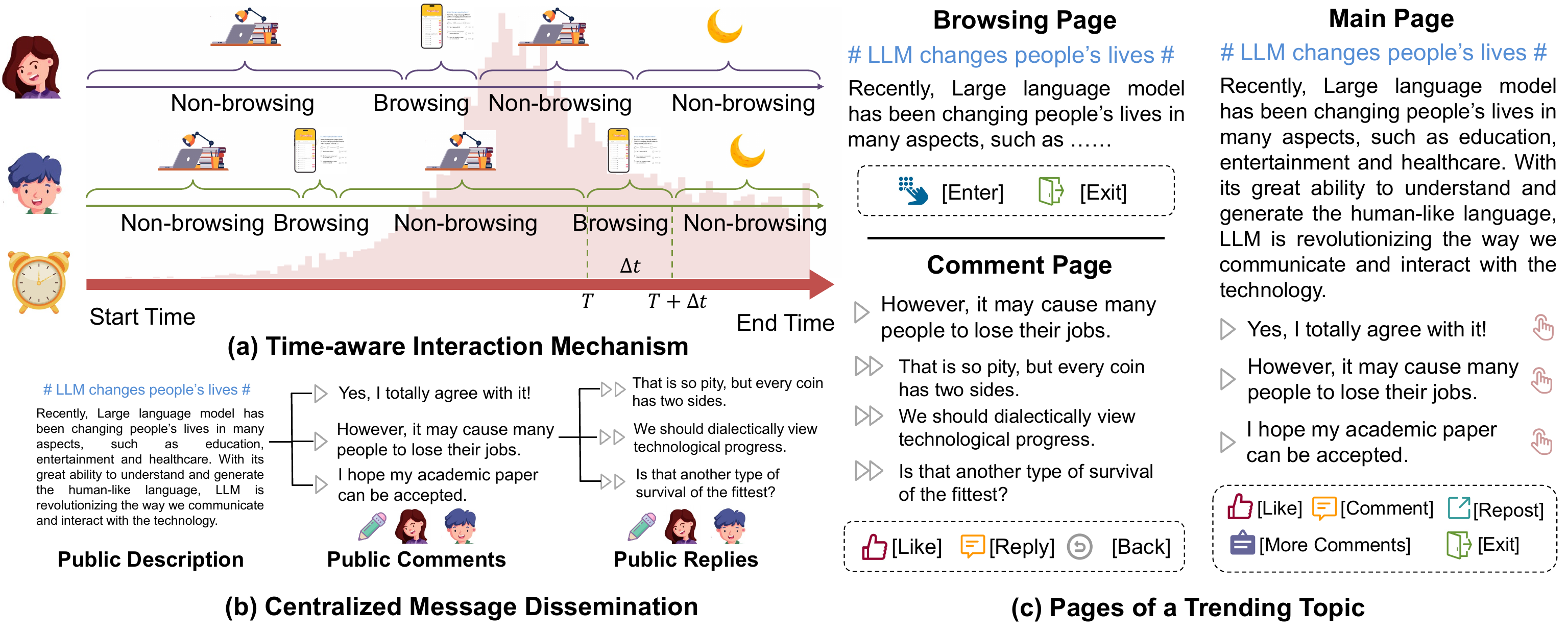}
		}
		\caption{The framework of TrendSim for simulating trending topics in social media.}
		\label{fig:framework}
		\vspace{-0.5cm}
	\end{figure*}
	
	\section{Related Works}
	\label{sec:related_work}
	
	\subsection{Poisoning Attack in Social Media}
	In recent years, poisoning attacks on social media platforms have gradually attracted widespread attention~\cite{khurana2019preventing}. It has been shown that social media serves as a primary way for spreading scams and malware, with numerous poisoning attacks occurring on social media platforms~\cite{kunwar2016social}.
	Several previous studies have investigated the intentions and characteristics of poisoning attackers in social media~\cite{aimeur2019manipulation, briscoe2014cues}.
	Their behaviors typically involve posting offensive comments and obscene images intended to propagate hate and perpetrate cyberbullying~\cite{chinivar2022online}.
	They also find that poisoning attacks are more inclined to focus on human vulnerabilities~\cite{aimeur2019manipulation}.
	
	Most previous studies focus on conventional contents in social networks, and pay less attention to trending topics in modern social media platforms. However, the poisoning attack on trending topics has become a critical problem, posing a substantial threat to our social environment, which demands wider attention from researchers.

	\subsection{LLM-based Multi-agent Social Simulation}
	Large language models have shown promising capabilities for building autonomous agents, attributed to their excellent language understanding and generation capabilities~\cite{ouyang2022training, park2023generative, wang2023large}.
	These agents are typically equipped with extensive modules based on LLMs to perform complex tasks~\cite{xi2023rise,shinn2023reflexion, zhu2023ghost, wang2023describe, qin2023toolllm}.
	Recently studies have proposed to utilize LLM-based agents for social simulations, in order to model human behaviors and interactions in different scenarios~\cite{wang2023large,park2023generative,gao2023s,gao2023large}. 
	For instance, $S^3$ ~\cite{gao2023s} simulates the emergence of social networks and phenomena like information diffusion through LLM-based agent interactions.
	RecAgent~\cite{wang2023large} proposes the simulation of user behaviors in the domain of recommender systems.
	Generative Agents~\cite{park2023generative} and AgentSims\cite{lin2023agentsims} create multi-agent systems as a digital town to replicate humans' daily lives. 
	
	However, previous frameworks fail in the scenario of trending topics in social media, mainly due to their lack of time consideration, centralized message dissemination, and reflection of user psychological conditions.
	Therefore, our work is the first one that simulates trending topics in social media and studies their poisoning attack problems.
	
	\section{Methods}
	\label{sec:framework}
	
	\subsection{Overview of TrendSim}
	TrendSim aims to simulate the complete lifecycle of a trending topic, interacted with users in the social media platform.
	Here, we formulate a trending topic as a public post that describes an explosive event, which users can interact by commenting, replying, and other actions.
	The lifecycle of a trending topic typically spans from its emergence to its disappearance, commonly lasting less than several hours.
	During this period, users can express their altitudes by commenting, or exchange opinions by replying to specific comments. These interactions push the evolution of a trending topic, raising a wide discussion, but they also provide opportunities for attackers to spread poisoning attacks.
	
	\subsection{Multi-agent Simulation Environment}
	\subsubsection{Time-aware Interaction Mechanism}
	Different from round-based simulation, our framework integrates a time-aware interaction mechanism, as illustrated in Figure~\ref{fig:framework}(a).
	Specifically, each session of interactions occurs at a specific timestamp $T$, and has a duration $\Delta t$. All the interactions execute in temporal order implemented with a dynamic priority queue~\cite{van1976design}.
	For example, if Alice views the trending topic at 3:12~PM and spends two minutes commenting, Bob can view Alice's comment at 3:14~PM.
	We assume that users access the trending topic following a certain probability distribution $P(t)$ with respect to time $t$.
	Specifically, $P(t)$ follows an exponential increase at the beginning, then experiences a growth deceleration before the peak, and finally takes a gradual decline to fade out~\cite{lerman2010using}. It also ensures $G_0$-smooth and $G_1$-smooth for continuity properties~\cite{barsky1989geometric}.
	Accordingly, we assume $P(t)$ as\\
	\resizebox{\linewidth}{!}
	{
		$P(t) \propto
		\left\{
		\begin{aligned}
			& e^{A(t-T_m)} & 0 \le t < T_m, \\
			& - \alpha A (t-T_m-\frac{1}{2\alpha})^2 + 1 + \frac{A}{4 \alpha} & T_m \le t < T_m + \frac{1}{\alpha}, \\
			& (t - T_m - \frac{1}{\alpha} + 1)^{-A} & t \ge T_m + \frac{1}{\alpha},
		\end{aligned}
		\right.$
	}
	where $A,T_m,\alpha,$ are hyper-parameters. Because of the page limitation, more details can be found in Appendix~\ref{appendix:entrance_prob} for better illustration.
	
	In order to improve the efficiency of simulation, we initially sample the first time of users' access based on the distribution before the simulation starts, and dynamically sample the next access time at the end of each access. 
	
	\subsubsection{Centralized Message Dissemination}
	Conventional social media messages typically spread through social networks based on user relationships. However, modern social media platforms commonly deploy an individual section to highlight trending topics, making their message dissemination through a centralized hub rather than a peer-to-peer network.
	For example, \textit{Weibo Hot Searches} provides a real-time list of top-50 popular hashtags as trending topics, and all the users can directly access them from the index page.
	Therefore, we implement centralized message dissemination for simulating trending topics.
	
	Specifically, TrendSim has three primary ways to disseminate messages, shown in Figure~\ref{fig:framework}(b).
	First of all, each trending topic features a public description consisting of a title, a summary, and the full content, which is visible to all the users in social media.
	Second, TrendSim allows users to post their comments under the trending topic, which are then accessible to others.
	Finally, TrendSim allows users to reply to any comments under the trending topic, and these replies are also visible to other users.
	By implementing these three mechanisms, users are able to get, post, and exchange their messages on trending topics in TrendSim.
	
	\subsubsection{User-Environment Interactive System}
	We design an interactive system between users and a trending topic. It defines the user's observation space, action space, and transition mechanism, illustrating how the user's action influences the trending topic.
	Specifically, users can be presented with three different pages of the trending topic, which are shown in Figure~\ref{fig:framework}(c):
	
	\noindent $\bullet$ \textit{Browsing Page}: users observe the title and summary of the trending topic, then take actions to view the details (navigate to Main Page) or leave the social media (finish this session).
	
	\noindent $\bullet$ \textit{Main Page}: users observe the title, full content, and top-$k$ comments. Then, users can choose one action among liking, commenting, reposting, viewing more comments (fetch next $k$ comments), viewing details of a specific comment (navigate to Comment Page) or leaving (finish this session). 
	
	\noindent $\bullet$ \textit{Comment Page}: users observe the specific comment with its top-$k$ replies, then take actions to like, reply or go back (navigate to Main Page).
	
	Moreover, user actions can influence the environment as well. Specifically, commenting on the trending topic or replying to comments can leave messages that influence subsequent observations of other users. Besides, liking a trending topic can add the popularity, and the numbers of likes on comments determine their rankings.
	In order to avoid taking infinite actions, we assume a maximum number of interactions in each session.

	\subsection{LLM-based User Agent}
	We design LLM-based agents to simulate users in social media. According to the human cognitive process~\cite{solso1979cognitive}, we design a perception, a memory, and an action module (see Figure~\ref{fig:user_agent}) to imitate human behaviors and reflect psychological conditions.
	\begin{figure*}[t]
		\centering
		\setlength{\fboxrule}{0.pt}
		\setlength{\fboxsep}{0.pt}
		\fbox{
			\includegraphics[width=0.98\linewidth]{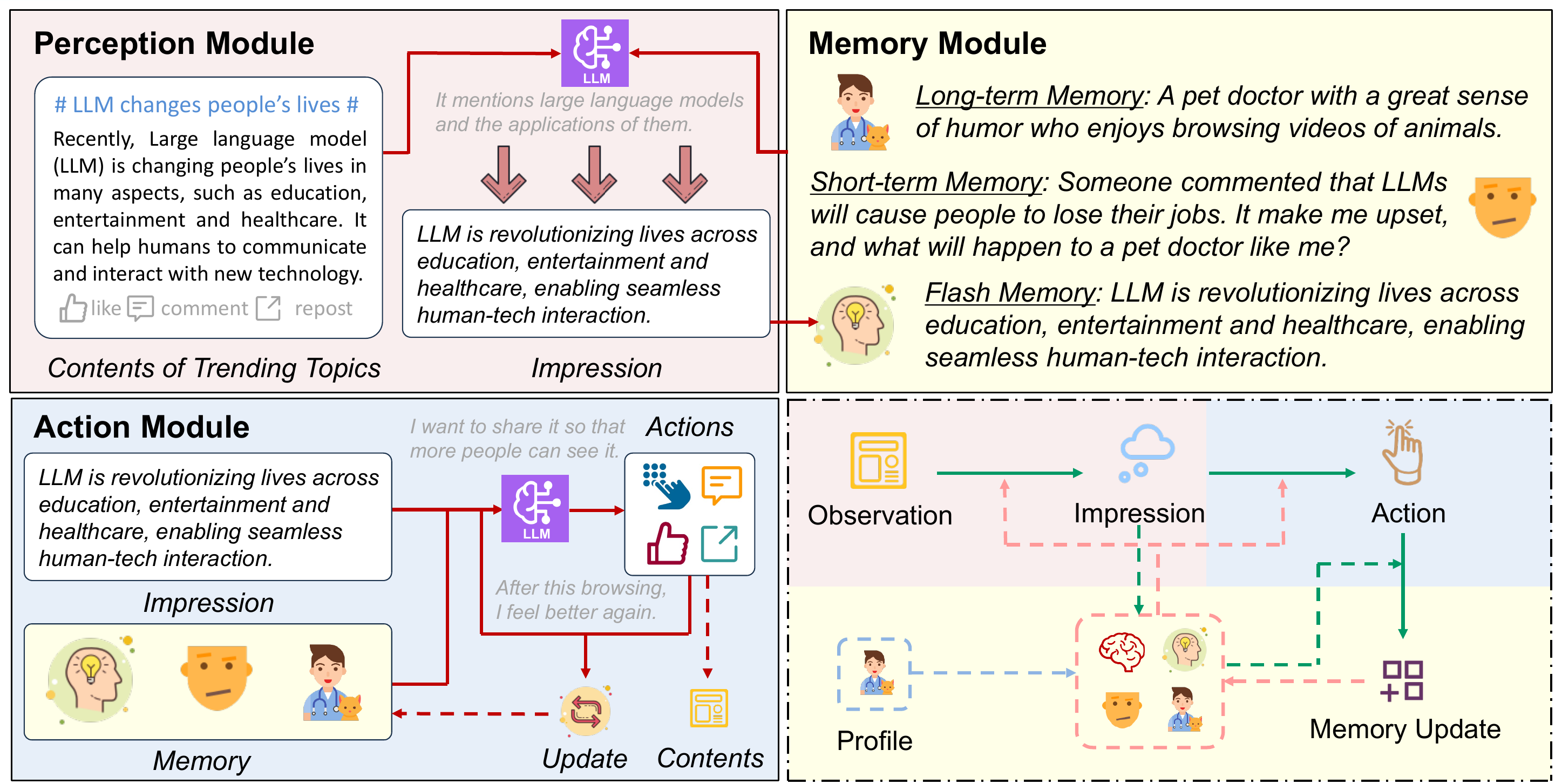}
		}
		\caption{The framework of LLM-based user agents in TrendSim.}
		\label{fig:user_agent}
		\vspace{-0.5cm}
	\end{figure*}

	\subsubsection{Perception Module}
	When browsing the contents of a trending topic, users always form impressions before their thinking and acting~\cite{solso1979cognitive}.
	Therefore, we design a perception module for agents to imitate this process.
	Specifically, we define the perception process as
	$$
	I \leftarrow Perception(f,O,M),
	$$
	where $I$ is the impression, $O$ represents the original observation, $M$ means the memory, and $f$ is implemented with an LLM.
	During this process, agents can express different attention in the generated impressions according to their memory, similar to diverse users in social media. For example, an optimistic user tends to focus on positive aspects, while a sad person who has just broken up is more likely to generate a negative impression.
	
	\subsubsection{Memory Module}
	\label{sec:memory_module}
	Memory is a significant component for human-like agents in social simulation, responsible for distinguishing one user from another~\cite{zhang2024survey}. It also dynamically affects user behaviors during the simulation.
	Therefore, according to cognitive psychology, we design a memory module consisting of three levels: long-term memory, short-term memory, and flash memory.
	
	Long-term memory incorporates the summary of the user's lifelong experiences (i.e., profiles), and remains unchanged during the simulation.
	Specifically, we implement long-term memory of agents by distilling their posts from real-world social media platforms before the simulation starts, followed by
	$$
	m_l \leftarrow LTM(f,\{p_1,p_2,...,p_N\}),
	$$
	where $m_l$ means the long-term memory, and $\{p_1,p_2,...,p_N\}$ are the posts of a real-world user.
	
	Short-term memory aims to maintain the dynamic condition of the user's psychology during the simulation.
	We utilize three major aspects to model short-term memory in our scenario as follows.
	
	\noindent $\bullet$ \textit{Emotion}: Direct emotional feeling of the user when browsing current contents.
	
	\noindent $\bullet$ \textit{Opinion}: Personal altitude of the user towards the trending topic after participating in it.
	
	\noindent $\bullet$ \textit{Social Confidence}: Personal belief of the user that trusts the justice of society.
	
	We design a reflection process to dynamically update short-term memory after each interaction. Specifically, we have
	$$
	m_s \leftarrow Reflection(f,I,A,m_s),
	$$
	where $m_s$ means short-term memory and $A$ indicates the action of user.
	Flash memory is the most immediate part of memory, storing the impression of observations. Specifically, we let flash memory $m_f=I$ for the current interaction. 
	In summary, the memory of an LLM-based user agent is defined as $M = (m_l,m_s,m_f)$ that affects user behaviors.

	\subsubsection{Action Module}
	Based on the memory, our user agents can generate actions towards their observations. Specifically, we implement the action module with LLMs by
	$$
	A \leftarrow Action(f,O,M).
	$$
	With the action module, user agents can take actions from the action space to affect the trending topics.
	As a result, users' messages can be shared among other users in social media.
	
	\subsection{Prototype-based Attacker Agent}
	\label{sec:attacker}
	In order to study poisoning attacks in trending topics, we develop prototype-based attacker agents to produce topic-specific malicious comments.
	Specifically, they generate poisoning comments based on a predefined prototype and the current observation with LLMs by
	$$
	\tilde{A} \leftarrow Action(f,P,O),
	$$
	where $\tilde{A}$ means the malicious comment, and $P$ is the prototype from a specific attacking target. 
	Moreover, we categorize three types of attackers according to their targets as follows.
	
	\subsubsection{Antisocial Attacker}
	Antisocial attackers aim to undermine users' social confidence by disseminating antisocial comments.
	For example, they often provoke conflicts between users and society to create discord.
	
	
	\subsubsection{Trolling Attacker}
	Trolling attackers intend to provoke, upset, or harass other users by posting offensive contents. They also make conflicts between different groups through disparagement, often targeting issues such as gender and values.
	
	\subsubsection{Rumor Attacker}
	Rumor attackers focus on generating and disseminating rumors about trending topics to obscure the truth. They deliberately foster misunderstandings about specific events and individuals.

	\section{Evaluations}
	\label{sec:alignment}
	
	\subsection{Simulation Settings}
	Based on TrendSim, we conduct simulations on trending topics in social media under poisoning attacks.
	We collect data from 1,000 public users from real-world social media platforms, anonymizing and summarizing their profiles from historical posts (see details in Appendix~\ref{appendix:details_data_collection}).
	We choose 10 trending topics, covering various domains and sentiment.
	We control the proportion of attackers among the total participants to simulate varying degrees of poisoning attacks.
	Due to the majority of Chinese corpus, we utilize GLM-3-turbo~\cite{du2022glm} as the foundation model $f$ in our simulation.
	For each trending topic, we assume the maximum lifecycle as 16 hours.
	
	We conduct evaluations of TrendSim from multiple aspects in this section, then collect and analyze the experimental results in Section~\ref{sec:experiment}.
	
	\subsection{Evaluation on User Agent}
	For user agents, we aim to evaluate their capability of acting as real-world users.
	Specifically, we utilize LLMs to score the consistency between characteristics and expressions of users on a scale from 0 to 1 (see Appendix~\ref{appendix:details_of_llm_eval_user_agent}).
	Our metrics include: (1) \textit{Behavior Consistency}: the consistency between characteristics and actions of users. (2) \textit{Psychology Consistency}: the consistency between characteristics and psychological conditions of users.
	We employ several LLMs as baselines by designing prompts to simulate users, and recruit a human expert as another baseline (see Appendix~\ref{appendix:details_of_baseline_user_agent}).
	
	\begin{table}[h]
		\centering
		\caption{Results of the evaluation on user agents.}
		\vspace{-0.2cm}
		
		\resizebox{\linewidth}{!}
		{
		\begin{tabular}{ccccc}
			\hline
			\hline
			\multirow{2}[4]{*}{\textbf{Methods}} & \multicolumn{4}{c}{\textbf{Behavior Consistency}} \bigstrut\\
			\cline{2-5}          & \textbf{GPT-4} & \textbf{GLM-4} & \textbf{Llama-3} & \textbf{Average} \bigstrut\\
			\hline
			GPT-4 & 0.925  & \underline{0.950}  & 0.830  & 0.902  \bigstrut[t]\\
			GLM-4 & 0.940  & \textbf{0.965} & \underline{0.842}  & \textbf{0.916 } \\
			Llama-3 & \underline{0.944}  & 0.930  & 0.826  & 0.900  \\
			TrendSim & \textbf{0.948 } & 0.945  & \textbf{0.853 } & \underline{0.915}  \\
			Human & 0.935  & \underline{0.950}  & 0.828  & 0.904  \bigstrut[b]\\
			\hline
			\multirow{2}[4]{*}{\textbf{Methods}} & \multicolumn{4}{c}{\textbf{Psychology Consistency}} \bigstrut\\
			\cline{2-5}          & \textbf{GPT-4} & \textbf{GLM-4} & \textbf{Llama-3} & \textbf{Average} \bigstrut\\
			\hline
			GPT-4 & 0.745  & \underline{0.767}  & 0.760  & 0.757  \bigstrut[t]\\
			GLM-4 & \underline{0.930}  & \textbf{0.776} & 0.767  & \underline{0.824}  \\
			Llama-3 & 0.705  & 0.640  & 0.768  & 0.704  \\
			TrendSim & \textbf{0.960 } & 0.745  & \underline{0.773}  & \textbf{0.826 } \\
			Human & 0.910  & 0.753  & \textbf{0.794 } & 0.819  \bigstrut[b]\\
			\hline
			\hline
		\end{tabular}
		}
		\vspace{-0.2cm}
		\label{tab:user_agent_alignment}%
	\end{table}%
	
	The results are shown in Table~\ref{tab:user_agent_alignment}, and we put the error bars in Appendix~\ref{appendix:user_agent_alignment_error_bar} due to the page limitation.
	We find that TrendSim achieves great performance in most cases, showing its capability of simulating users in trending topics.
	We also observe that humans may not do well in playing the roles of others in our scenario.
	However, the results among different LLM evaluators can be various.
	
	\subsection{Evaluation on Attacker Agent}
	For attacker agents, we evaluate their generated poisoning comments on two aspects: (1) \textit{Consistency}: the relevance between comments and viewed contents. (2) \textit{Concealment}: the ability to evade malicious detection. We show details in Appendix~\ref{appendix:details_of_llm_eval_attacker_agent}.
	We also utilize LLMs and human as baselines (see Appendix~\ref{appendix:details_of_baseline_attacker_agent}).
	The results are shown in Table~\ref{tab:attacker_alignment}, and we put the error bars in Appendix~\ref{appendix:attacker_agent_alignment_error_bar}.
	
	\begin{table}[h]
		\centering
		\caption{Results of the evaluation on attacker agents.}
		\vspace{-0.2cm}
		\resizebox{\linewidth}{!}
			{
				\begin{tabular}{ccccc}
					\hline
					\hline
					\multirow{2}[4]{*}{\textbf{Methods}} & \multicolumn{4}{c}{\textbf{Consistency}} \bigstrut\\
					\cline{2-5}          & \textbf{GPT-4} & \textbf{GLM-4} & \textbf{Llama-3} & \textbf{Average} \bigstrut\\
					\hline
					GPT-4 & 0.368  & \underline{0.675}  & 0.744  & 0.596  \bigstrut[t]\\
					GLM-4 & 0.435  & 0.644  & 0.755  & 0.611  \\
					Llama-3 & 0.320  & 0.475  & 0.735  & 0.510  \\
					TrendSim & \textbf{0.815 } & \textbf{0.905 } & \textbf{0.792 } & \textbf{0.837 } \\
					Human & \underline{0.575}  & 0.540  & \underline{0.784}  & \underline{0.633}  \bigstrut[b]\\
					\hline
					\multirow{2}[4]{*}{\textbf{Methods}} & \multicolumn{4}{c}{\textbf{Concealment}} \bigstrut\\
					\cline{2-5}          & \textbf{GPT-4} & \textbf{GLM-4} & \textbf{Llama-3} & \textbf{Average} \bigstrut\\
					\hline
					GPT-4 & 0.342  & 0.380  & 0.325  & 0.349  \bigstrut[t]\\
					GLM-4 & 0.445  & 0.370  & 0.317  & 0.377  \\
					Llama-3 & 0.475  & 0.370  & 0.291  & 0.379  \\
					TrendSim & \textbf{0.770 } & \underline{0.449}  & \underline{0.335}  & \textbf{0.518 } \\
					Human & \underline{0.756}  & \textbf{0.453 } & \textbf{0.340 } & \underline{0.516}  \bigstrut[b]\\
					\hline
					\hline
				\end{tabular}
			}
			\vspace{-0.5cm}
			\label{tab:attacker_alignment}%
		\end{table}%
	
	We find that our prototype-based attackers outperform in most cases, indicating the effectiveness of our mechanism. Moreover, it also shows that the poisoning attacks from LLMs and agents can be more severe than human attackers, and harder to be detected as well.

	\subsection{Evaluation on Multi-agent System}
	Besides evaluating user agents and attacker agents in the single-agent view, we also evaluate the multi-agent system from the interactive system perspective.
	We focus on two metrics: (1) \textit{Rationality}: the rationality of discussions from the comments. (2) \textit{Diversity}: the distinction among different users. Details are in Appendix~\ref{appendix:details_of_llm_eval_system}.
	The results are shown in Table~\ref{tab:system_alignment}, with the error bars in Appendix~\ref{appendix:system_alignment}. The results show that TrendSim has a great performance, where most of rationality scores are above 0.8 and most diversity scores are above 0.7.
	
	\begin{table}[h]
		\centering
		\caption{Results of the evaluation on multi-agent system.}
		\vspace{-0.2cm}
			\resizebox{\linewidth}{!}
			{
				\begin{tabular}{ccccc}
				    \hline
				    \hline
				    \multirow{2}[4]{*}{\textbf{Sentiment}} & \multicolumn{4}{c}{\textbf{Rationality}} \bigstrut\\
				    \cline{2-5}          & \textbf{GPT-4} & \textbf{GLM-4} & \textbf{Llama-3} & \textbf{Average} \bigstrut\\
				    \hline
				    Positive & 0.865  & 0.830  & 0.840  & 0.845  \bigstrut[t]\\
				    Negative & 0.775  & 0.725  & 0.775  & 0.758  \\
				    Neutral & 0.815  & 0.815  & 0.817  & 0.816  \\
				    All   & 0.818  & 0.790  & 0.811  & 0.806  \\
				    \hline
				    \multirow{2}[3]{*}{\textbf{Sentiment}} & \multicolumn{4}{c}{\textbf{Diversity}} \bigstrut[b]\\
				    \cline{2-5}          & \textbf{GPT-4} & \textbf{GLM-4} & \textbf{Llama-3} & \textbf{Average} \bigstrut\\
				    \hline
				    Positive & 0.785  & 0.770  & 0.835  & 0.797  \bigstrut[t]\\
				    Negative & 0.746  & 0.760  & 0.805  & 0.770  \\
				    Neutral & 0.685  & 0.765  & 0.744  & 0.731  \\
				    All   & 0.739  & 0.765  & 0.795  & 0.766  \bigstrut[b]\\
				    \hline
				    \hline
			    \end{tabular}%
			}
			\vspace{-0.4cm}
			\label{tab:system_alignment}%
		\end{table}%
	
	\subsection{Evaluation on Simulation Efficiency}
	The efficiency of simulations is crucial for researchers because increased efficiency leads to reduced time costs and enhances the capability to scale to larger user bases.
	Therefore, we evaluate the time cost of our simulations on TrendSim.
	All the experiments are conducted under \textit{Intel\textsuperscript{\textregistered} Xeon\textsuperscript{\textregistered} Gold-5118 (48 Core)} CPU and GLM-3-turbo API.
	The results are shown in Table~\ref{tab:time_cost}, and we find that our simulations take around 16 hours to simulate a trending topic with 1,000 participants in social media.
	It is promising to further improve the efficiency with parallel techniques in future works.
	
	\begin{table}[t]
		\centering
		\caption{The reference of time cost with different numbers of participant agents. SE, PA-10, PA-30, and PA-50 refer to the simulation with 0\%, 10\%, 30\% and 50\% attacker agents respectively.}
		\vspace{-0.2cm}
		\resizebox{\linewidth}{!}
		{
		\begin{tabular}{ccccccc}    
			\hline
			\hline
			\multirow{2}[4]{*}{\textbf{Degree}} & \multicolumn{6}{c}{\textbf{Time Cost (hours)}} \bigstrut\\
			\cline{2-7}          & \textbf{\# 10} & \textbf{\# 50} & \textbf{\# 100} & \textbf{\# 200} & \textbf{\# 500} & \textbf{\# 1000} \bigstrut\\
			\hline
			SE    & 0.2   & 0.7   & 1.6   & 3.3   & 6.5   & 16 \bigstrut[t]\\
			PA-10 & 0.1   & 0.6   & 1.5   & 2.8   & 6.1   & 15 \\
			PA-30 & 0.8 $\times 10^{-1}$   & 0.4   & 1.2   & 2.2   & 4.9   & 12 \\
			PA-50 & 0.6 $\times 10^{-1}$   & 0.3   & 1.0   & 1.7   & 4.0   & 8.0 \bigstrut[b]\\
			\hline
			\hline
		\end{tabular}%
		}
		\vspace{-0.4cm}
		\label{tab:time_cost}%
	\end{table}%

	\section{Experiments}
	\label{sec:experiment}
	Based on TrendSim, we conduct simulation experiments and analyze their results.
	We focus on four critical problems about poisoning attacks of trending topics in social media.

	\noindent \textbf{Problem 1: What negative impact do poisoning attacks have on trending topics?}
	
	\begin{table*}[!h]
		\centering
		\caption{Results of the negative impact of poisoning attacks on trending topics in simulation experiments. \textit{Average} indicates the mean value of all user conditions in single trending topic, and \textit{Divergence} means the standard deviation of all user conditions in single trending topic. Their means and standard deviations~(denoted by $\pm$) are calculated across all trending topics of that group.}
		
		\resizebox{\textwidth}{!}
		{
			\begin{tabular}{>{\centering\arraybackslash}p{2.0cm}>{\centering\arraybackslash}p{1.6cm}>{\centering\arraybackslash}p{2.8cm}>{\centering\arraybackslash}p{2.8cm}>{\centering\arraybackslash}p{2.8cm}>{\centering\arraybackslash}p{2.8cm}}
				\hline
				\hline
				\multirow{2}[4]{*}{\textbf{Groups}} & \multirow{2}[4]{*}{\textbf{Degrees}} & \multicolumn{2}{c}{\textbf{Emotion}} & \multicolumn{2}{c}{\textbf{Social Confidence}} \bigstrut\\
				\cline{3-6}      &       & \textbf{Average} & \textbf{Divergence} & \textbf{Average} & \textbf{Divergence} \bigstrut\\
				\hline
				\multirow{4}[2]{*}{Positive} & SE    & 0.886$\pm$0.057 & 0.140$\pm$0.040 & 0.905$\pm$0.040 & 0.109$\pm$0.031 \bigstrut[t]\\
				& PA-10 & 0.812$\pm$0.145 & 0.151$\pm$0.028 & 0.836$\pm$0.126 & 0.132$\pm$0.034 \\
				& PA-30 & 0.819$\pm$0.081 & 0.180$\pm$0.024 & 0.845$\pm$0.068 & 0.152$\pm$0.030 \\
				& PA-50 & 0.813$\pm$0.048 & 0.190$\pm$0.017 & 0.836$\pm$0.035 & 0.168$\pm$0.015 \bigstrut[b]\\
				\hline
				\multirow{4}[2]{*}{Negative} & SE    & 0.443$\pm$0.002 & 0.065$\pm$0.011 & 0.457$\pm$0.004 & 0.078$\pm$0.011 \bigstrut[t]\\
				& PA-10 & 0.429$\pm$0.000 & 0.066$\pm$0.005 & 0.441$\pm$0.009 & 0.081$\pm$0.005 \\
				& PA-30 & 0.430$\pm$0.008 & 0.067$\pm$0.009 & 0.443$\pm$0.020 & 0.084$\pm$0.009 \\
				& PA-50 & 0.442$\pm$0.007 & 0.071$\pm$0.002 & 0.457$\pm$0.023 & 0.083$\pm$0.004 \bigstrut[b]\\
				\hline
				\multirow{4}[2]{*}{Netural} & SE    & 0.525$\pm$0.083 & 0.120$\pm$0.056 & 0.570$\pm$0.122 & 0.123$\pm$0.049 \bigstrut[t]\\
				& PA-10 & 0.509$\pm$0.078 & 0.114$\pm$0.057 & 0.550$\pm$0.117 & 0.120$\pm$0.049 \\
				& PA-30 & 0.509$\pm$0.073 & 0.117$\pm$0.053 & 0.552$\pm$0.110 & 0.121$\pm$0.048 \\
				& PA-50 & 0.490$\pm$0.058 & 0.104$\pm$0.048 & 0.523$\pm$0.091 & 0.117$\pm$0.049 \bigstrut[b]\\
				\hline
				\multirow{4}[2]{*}{All} & SE    & 0.653$\pm$0.203 & 0.117$\pm$0.052 & 0.681$\pm$0.204 & 0.109$\pm$0.041 \bigstrut[t]\\
				& PA-10 & 0.614$\pm$0.194 & 0.119$\pm$0.051 & 0.643$\pm$0.196 & 0.117$\pm$0.042 \\
				& PA-30 & 0.617$\pm$0.181 & 0.132$\pm$0.057 & 0.647$\pm$0.185 & 0.126$\pm$0.044 \\
				& PA-50 & 0.610$\pm$0.174 & 0.131$\pm$0.059 & 0.635$\pm$0.177 & 0.131$\pm$0.046 \bigstrut[b]\\
				\hline
				\hline
			\end{tabular}
		}
		\label{tab:issue_01}%
	\end{table*}%
	
	As we have discussed above, poisoning attacks can lead to negative impacts through trending topics in social media.
	Therefore, we intend to verify and analyze this phenomenon with TrendSim.
	Specifically, we study the conditions of the user's psychology during the simulation, and compare them among four levels of attacks.
	The metrics of psychological conditions include \textit{Emotion} and \textit{Social Confidence} that we have introduced in Section~\ref{sec:memory_module}, expressed ranging from 0 to 1.
	The levels of attacks include:
	(1) \textit{SE}: no attackers injected.
	(2) \textit{PA-10}: 10\% attackers injected.
	(3) \textit{PA-30}: 30\% attackers injected.
	(4) \textit{PA-50}: 50\% attackers injected.
	These attackers are introduced according to Section~\ref{sec:attacker}.
	In addition, we categorize the trending topics into three sentiment groups, including positive, negative, and neutral.
	
	We focus at the end time of simulations, calculating the average conditions and divergences among users.
	The results are shown in Table~\ref{tab:issue_01}.
	We find that most poisoning attacks have negative impacts on users in social media, but their effects vary in different groups.
	Compared with the negative and neutral groups, the positive group is affected the most, potentially due to the larger contrast between poisoning comments and normal comments.
	Moreover, we find that the results are not proportional to the levels of attacks, where fewer attackers can probably cause larger impacts.

	\noindent \textbf{Problem 2: How do users' psychological conditions dynamically change over time?}
	
	We further study the dynamic changes in users' psychological conditions over time.
	Specifically, we track the psychological conditions of users along the timeline, drawing the curves in Figure~\ref{fig:attack_overtime}.
	For better demonstration, we show the results relative to that in SE.
	We find that a sharp decrease happens in the middle of time, which is also the moment when a large number of users enter simultaneously.
	Moreover, PA-50 exhibits the most negative impacts throughout the entire process.
	Due to the page limitation, we show the results of different groups in Appendix~\ref{appendix:issue_02_group}.
	
	\begin{figure}[t]
		\centering
		\begin{subfigure}[b]{0.48\linewidth}
			\includegraphics[width=\linewidth]{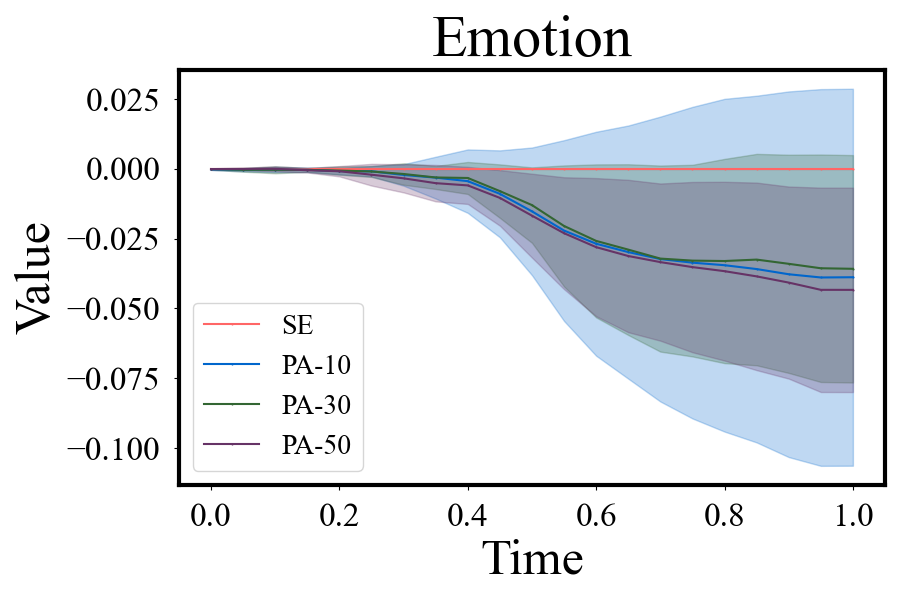}
		\end{subfigure}
		\hfil
		\begin{subfigure}[b]{0.48\linewidth}
			\includegraphics[width=\linewidth]{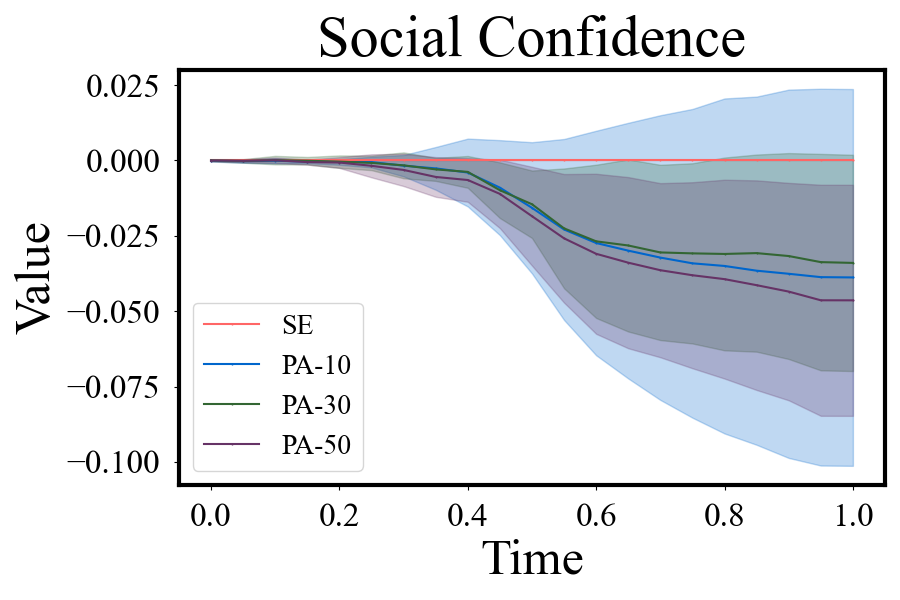}
		\end{subfigure}
		\caption{The user's psychological conditions in social media platforms over time. The curves represent mean values, and the shading represents standard deviations.}
		\vspace{-0.5cm}
		\label{fig:attack_overtime}
	\end{figure}

	\noindent \textbf{Problem 3: What types of users are more susceptible to poisoning attacks in trending topics?}
	
	We further study what types of users are more susceptible to poisoning attacks.
	Specifically, we categorize all these 1,000 users into six groups according to their preferences, including  \textit{Entertainment}, \textit{Sports}, \textit{Lifestyle}, \textit{Society}, \textit{Culture} and \textit{Technology}.
	We present the proportion of each type in Figure~\ref{fig:issue_03}(a). It shows that the entertainment group and society group are two dominant groups.
	
	The simulation results are shown in Figure~\ref{fig:issue_03}(b) and Figure~\ref{fig:issue_03}(c).
	We find that people who are interested in social topics are most susceptible to poisoning attacks, which aligns with our intuition. However, to our surprise, we find that the entertainment group suffers the least impact on poisoning attacks, compared with other groups.

	\begin{figure*}[t]
		\centering
		\begin{subfigure}[b]{0.32\textwidth}
			\includegraphics[width=\textwidth]{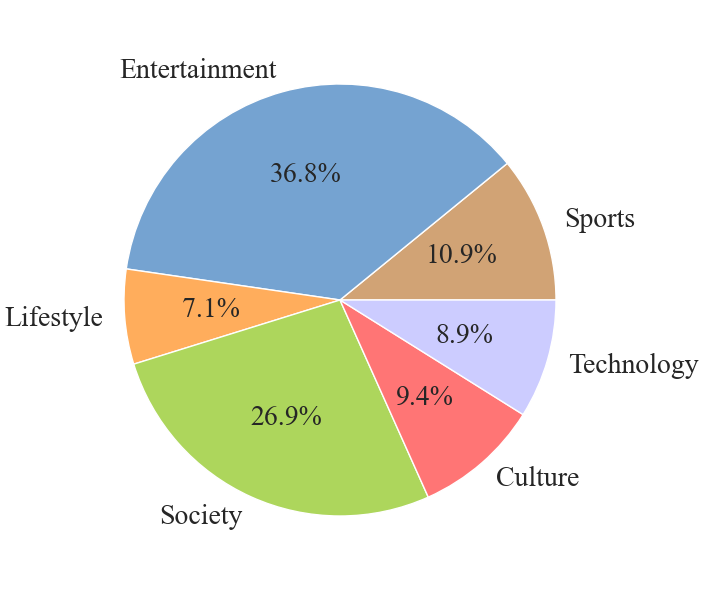}
			\caption{Distribution of user groups.}
			\label{fig:user_group}
		\end{subfigure}
		\begin{subfigure}[b]{0.32\textwidth}
			\includegraphics[width=\textwidth]{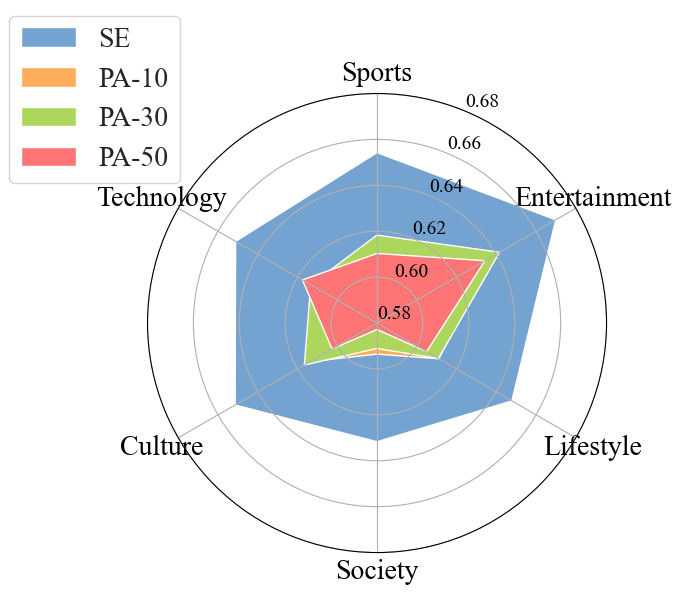}
			\caption{Emotion.}
			\label{fig:exp03_em}
		\end{subfigure}
		\begin{subfigure}[b]{0.32\textwidth}
			\includegraphics[width=\textwidth]{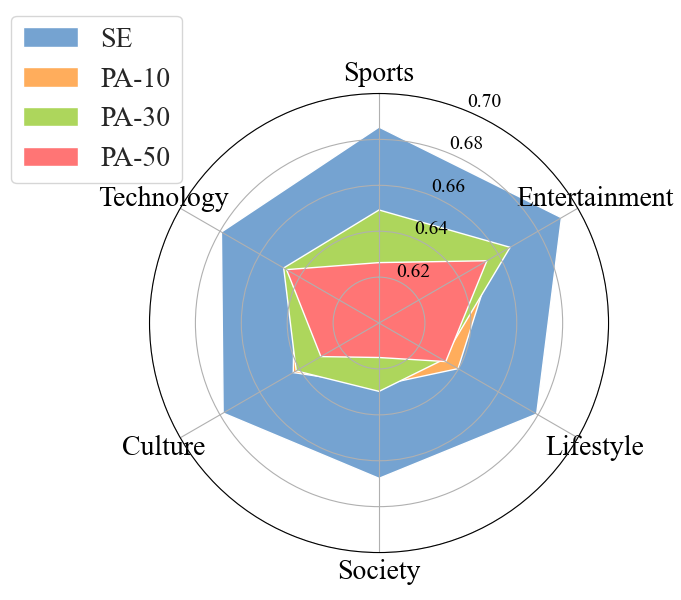}
			\caption{Social Confidence.}
			\label{fig:exp03_sc}
		\end{subfigure}
		\caption{The user psychological conditions in different groups of simulation experiments.}
		\vspace{-0.4cm}
		\label{fig:issue_03}
	\end{figure*}
	
	\noindent \textbf{Problem 4: How effectively the content censorship defend against poisoning attacks?}
	
	In order to mitigate the negative impact of poisoning attacks, social media platforms commonly devise defensive strategies against malicious contents.
	One of the most common methods for protection is content censorship, which aims to detect and filter poisoning comments.
	In this part, we further conduct simulation experiments to assess the effectiveness of content censorship in defending against poisoning attacks.
	Specifically, we run PA-50 simulations with the content censorship implemented by LLMs.
	The results are shown in Appendix~\ref{appendix:issue_04}. We find that the content censorship mechanism can effectively mitigate the negative impact of poisoning attacks in most cases.

	\section{Conclusion}
	\label{sec:conclusion}
	
	In this paper, we propose an LLM-based multi-agent system to simulate trending topics in social media under poisoning attacks, named TrendSim.
	By implementing the time-aware interaction mechanism, the centralized message dissemination, and an interactive system, we develop an environment for simulating trending topics.
	We also propose LLM-based human-like agents to imitate users, and design prototype-based attackers to simulate poisoning attacks.
	Our evaluations show the effectiveness and efficiency of TrendSim.
	Based on it, we conduct simulations to study four critical problems about poisoning attacks on trending topics, and draw conclusions of them.
	
	We hope our work can make contributions to society, thereby achieving a warm and responsible usage of artificial intelligence.
	In future works, it is promising to focus on the memory mechanism of LLM-based agents for social simulations, which is a critical part of role-playing and personalization.
        It is also be valuable to implement larger-scale social simulations for various applications.
	
	\section*{Limitations}
        In this work, we propose TrendSim to simulate trending topics in social media under poisoning attacks.
        However, there are still certain limitations.
        First of all, TrendSim currently conducts simulations solely in textual form, without multi-modal information that is important in social media as well. For instance, fake photos can also lead to poisoning attacks in trending topics.
        Second, like previous works on social simulations, TrendSim is also based on a series of assumptions. These assumptions arise from unobserved factors in the real world that researchers cannot fully account for.
        Consequently, due to these assumptions, some biases may exist in our simulations, limiting their ability to accurately replicate all real-world details. This alignment problem also occurs in other applications of social simulations.

	\section*{Ethical Impacts}
 	TrendSim aids in uncovering insights into poisoning attacks on social media trends, which is beneficial for developing defense mechanisms.
	However, every coin has two sides. The availability of TrendSim may inadvertently contribute to the evolution of attackers, who could adapt and evolve in response to the defense mechanisms employed on social media platforms.

\bibliography{custom}

\begin{thebibliography}{26}
\providecommand{\natexlab}[1]{#1}

\bibitem[{A{\"\i}meur et~al.(2019)A{\"\i}meur, D{\'\i}az~Ferreyra, and
  Hage}]{aimeur2019manipulation}
Esma A{\"\i}meur, Nicol{\'a}s D{\'\i}az~Ferreyra, and Hicham Hage. 2019.
\newblock Manipulation and malicious personalization: exploring the
  self-disclosure biases exploited by deceptive attackers on social media.
\newblock \emph{Frontiers in artificial intelligence}, 2:26.

\bibitem[{Barsky and DeRose(1989)}]{barsky1989geometric}
Brian~A Barsky and Tony~D DeRose. 1989.
\newblock Geometric continuity of parametric curves: three equivalent
  characterizations.
\newblock \emph{IEEE Computer Graphics and Applications}, 9(6):60--69.

\bibitem[{Briscoe et~al.(2014)Briscoe, Appling, and Hayes}]{briscoe2014cues}
Erica~J Briscoe, D~Scott Appling, and Heather Hayes. 2014.
\newblock Cues to deception in social media communications.
\newblock In \emph{2014 47th Hawaii international conference on system
  sciences}, pages 1435--1443. IEEE.

\bibitem[{Chinivar et~al.(2022)Chinivar, Roopa, Arunalatha, and
  Venugopal}]{chinivar2022online}
Sneha Chinivar, MS~Roopa, JS~Arunalatha, and KR~Venugopal. 2022.
\newblock Online offensive behaviour in socialmedia: Detection approaches,
  comprehensive review and future directions.
\newblock \emph{Entertainment Computing}, page 100544.

\bibitem[{Du et~al.(2022)Du, Qian, Liu, Ding, Qiu, Yang, and Tang}]{du2022glm}
Zhengxiao Du, Yujie Qian, Xiao Liu, Ming Ding, Jiezhong Qiu, Zhilin Yang, and
  Jie Tang. 2022.
\newblock Glm: General language model pretraining with autoregressive blank
  infilling.
\newblock In \emph{Proceedings of the 60th Annual Meeting of the Association
  for Computational Linguistics (Volume 1: Long Papers)}, pages 320--335.

\bibitem[{Gao et~al.(2023{\natexlab{a}})Gao, Lan, Li, Yuan, Ding, Zhou, Xu, and
  Li}]{gao2023large}
Chen Gao, Xiaochong Lan, Nian Li, Yuan Yuan, Jingtao Ding, Zhilun Zhou, Fengli
  Xu, and Yong Li. 2023{\natexlab{a}}.
\newblock Large language models empowered agent-based modeling and simulation:
  A survey and perspectives.
\newblock \emph{arXiv preprint arXiv:2312.11970}.

\bibitem[{Gao et~al.(2023{\natexlab{b}})Gao, Lan, Lu, Mao, Piao, Wang, Jin, and
  Li}]{gao2023s}
Chen Gao, Xiaochong Lan, Zhihong Lu, Jinzhu Mao, Jinghua Piao, Huandong Wang,
  Depeng Jin, and Yong Li. 2023{\natexlab{b}}.
\newblock S$^3$: Social-network simulation system with large language
  model-empowered agents.
\newblock \emph{arXiv preprint arXiv:2307.14984}.

\bibitem[{Hua et~al.(2023)Hua, Fan, Li, Mei, Ji, Ge, Hemphill, and
  Zhang}]{hua2023war}
Wenyue Hua, Lizhou Fan, Lingyao Li, Kai Mei, Jianchao Ji, Yingqiang Ge, Libby
  Hemphill, and Yongfeng Zhang. 2023.
\newblock War and peace (waragent): Large language model-based multi-agent
  simulation of world wars.
\newblock \emph{arXiv preprint arXiv:2311.17227}.

\bibitem[{Khurana et~al.(2019)Khurana, Mittal, Piplai, and
  Joshi}]{khurana2019preventing}
Nitika Khurana, Sudip Mittal, Aritran Piplai, and Anupam Joshi. 2019.
\newblock Preventing poisoning attacks on ai based threat intelligence systems.
\newblock In \emph{2019 IEEE 29th International Workshop on Machine Learning
  for Signal Processing (MLSP)}, pages 1--6. IEEE.

\bibitem[{Kova{\v{c}} et~al.(2023)Kova{\v{c}}, Portelas, Dominey, and
  Oudeyer}]{kovavc2023socialai}
Grgur Kova{\v{c}}, R{\'e}my Portelas, Peter~Ford Dominey, and Pierre-Yves
  Oudeyer. 2023.
\newblock The socialai school: Insights from developmental psychology towards
  artificial socio-cultural agents.
\newblock \emph{arXiv preprint arXiv:2307.07871}.

\bibitem[{Kunwar and Sharma(2016)}]{kunwar2016social}
Rakesh~Singh Kunwar and Priyanka Sharma. 2016.
\newblock Social media: A new vector for cyber attack.
\newblock In \emph{2016 International Conference on Advances in Computing,
  Communication, \& Automation (ICACCA)(Spring)}, pages 1--5. IEEE.

\bibitem[{Lerman and Hogg(2010)}]{lerman2010using}
Kristina Lerman and Tad Hogg. 2010.
\newblock Using a model of social dynamics to predict popularity of news.
\newblock In \emph{Proceedings of the 19th international conference on World
  wide web}, pages 621--630.

\bibitem[{Lin et~al.(2023)Lin, Zhao, Zhang, Wu, Ping, and
  Chen}]{lin2023agentsims}
Jiaju Lin, Haoran Zhao, Aochi Zhang, Yiting Wu, Huqiuyue Ping, and Qin Chen.
  2023.
\newblock Agentsims: An open-source sandbox for large language model
  evaluation.
\newblock \emph{arXiv preprint arXiv:2308.04026}.

\bibitem[{Ouyang et~al.(2022)Ouyang, Wu, Jiang, Almeida, Wainwright, Mishkin,
  Zhang, Agarwal, Slama, Ray et~al.}]{ouyang2022training}
Long Ouyang, Jeffrey Wu, Xu~Jiang, Diogo Almeida, Carroll Wainwright, Pamela
  Mishkin, Chong Zhang, Sandhini Agarwal, Katarina Slama, Alex Ray, et~al.
  2022.
\newblock Training language models to follow instructions with human feedback.
\newblock \emph{Advances in Neural Information Processing Systems},
  35:27730--27744.

\bibitem[{Park et~al.(2023)Park, O'Brien, Cai, Morris, Liang, and
  Bernstein}]{park2023generative}
Joon~Sung Park, Joseph~C. O'Brien, Carrie~J. Cai, Meredith~Ringel Morris, Percy
  Liang, and Michael~S. Bernstein. 2023.
\newblock Generative agents: Interactive simulacra of human behavior.
\newblock In \emph{In the 36th Annual ACM Symposium on User Interface Software
  and Technology (UIST '23}, UIST '23, New York, NY, USA. Association for
  Computing Machinery.

\bibitem[{Qin et~al.(2023)Qin, Liang, Ye, Zhu, Yan, Lu, Lin, Cong, Tang, Qian
  et~al.}]{qin2023toolllm}
Yujia Qin, Shihao Liang, Yining Ye, Kunlun Zhu, Lan Yan, Yaxi Lu, Yankai Lin,
  Xin Cong, Xiangru Tang, Bill Qian, et~al. 2023.
\newblock {ToolLLM}: Facilitating large language models to master 16000+
  real-world apis.
\newblock \emph{arXiv preprint arXiv:2307.16789}.

\bibitem[{Shinn et~al.(2023)Shinn, Cassano, Labash, Gopinath, Narasimhan, and
  Yao}]{shinn2023reflexion}
Noah Shinn, Federico Cassano, Beck Labash, Ashwin Gopinath, Karthik Narasimhan,
  and Shunyu Yao. 2023.
\newblock Reflexion: Language agents with verbal reinforcement learning.
\newblock \emph{arXiv preprint arXiv:2303.11366}.

\bibitem[{Solso and Kagan(1979)}]{solso1979cognitive}
Robert~L Solso and Jerome Kagan. 1979.
\newblock \emph{Cognitive psychology}.
\newblock Houghton Mifflin Harcourt P.

\bibitem[{van Emde~Boas et~al.(1976)van Emde~Boas, Kaas, and
  Zijlstra}]{van1976design}
Peter van Emde~Boas, Robert Kaas, and Erik Zijlstra. 1976.
\newblock Design and implementation of an efficient priority queue.
\newblock \emph{Mathematical systems theory}, 10(1):99--127.

\bibitem[{Wang et~al.(2023{\natexlab{a}})Wang, Ma, Feng, Zhang, Yang, Zhang,
  Chen, Tang, Chen, Lin et~al.}]{wang2023survey}
Lei Wang, Chen Ma, Xueyang Feng, Zeyu Zhang, Hao Yang, Jingsen Zhang, Zhiyuan
  Chen, Jiakai Tang, Xu~Chen, Yankai Lin, et~al. 2023{\natexlab{a}}.
\newblock A survey on large language model based autonomous agents.
\newblock \emph{arXiv preprint arXiv:2308.11432}.

\bibitem[{Wang et~al.(2023{\natexlab{b}})Wang, Zhang, Yang, Chen, Tang, Zhang,
  Chen, Lin, Song, Zhao, Xu, Dou, Wang, and Wen}]{wang2023large}
Lei Wang, Jingsen Zhang, Hao Yang, Zhiyuan Chen, Jiakai Tang, Zeyu Zhang,
  Xu~Chen, Yankai Lin, Ruihua Song, Wayne~Xin Zhao, Jun Xu, Zhicheng Dou, Jun
  Wang, and Ji-Rong Wen. 2023{\natexlab{b}}.
\newblock \href {https://arxiv.org/abs/2306.02552} {When large language model
  based agent meets user behavior analysis: A novel user simulation paradigm}.
\newblock \emph{Preprint}, arXiv:2306.02552.

\bibitem[{Wang et~al.(2023{\natexlab{c}})Wang, Cai, Liu, Ma, and
  Liang}]{wang2023describe}
Zihao Wang, Shaofei Cai, Anji Liu, Xiaojian Ma, and Yitao Liang.
  2023{\natexlab{c}}.
\newblock Describe, explain, plan and select: Interactive planning with large
  language models enables open-world multi-task agents.
\newblock \emph{arXiv preprint arXiv:2302.01560}.

\bibitem[{Xi et~al.(2023)Xi, Chen, Guo, He, Ding, Hong, Zhang, Wang, Jin, Zhou
  et~al.}]{xi2023rise}
Zhiheng Xi, Wenxiang Chen, Xin Guo, Wei He, Yiwen Ding, Boyang Hong, Ming
  Zhang, Junzhe Wang, Senjie Jin, Enyu Zhou, et~al. 2023.
\newblock The rise and potential of large language model based agents: A
  survey.
\newblock \emph{arXiv preprint arXiv:2309.07864}.

\bibitem[{Zhang et~al.(2024)Zhang, Bo, Ma, Li, Chen, Dai, Zhu, Dong, and
  Wen}]{zhang2024survey}
Zeyu Zhang, Xiaohe Bo, Chen Ma, Rui Li, Xu~Chen, Quanyu Dai, Jieming Zhu,
  Zhenhua Dong, and Ji-Rong Wen. 2024.
\newblock A survey on the memory mechanism of large language model based
  agents.
\newblock \emph{arXiv preprint arXiv:2404.13501}.

\bibitem[{Zhao et~al.(2023)Zhao, Zhou, Li, Tang, Wang, Hou, Min, Zhang, Zhang,
  Dong et~al.}]{zhao2023survey}
Wayne~Xin Zhao, Kun Zhou, Junyi Li, Tianyi Tang, Xiaolei Wang, Yupeng Hou,
  Yingqian Min, Beichen Zhang, Junjie Zhang, Zican Dong, et~al. 2023.
\newblock A survey of large language models.
\newblock \emph{arXiv preprint arXiv:2303.18223}.

\bibitem[{Zhu et~al.(2023)Zhu, Chen, Tian, Tao, Su, Yang, Huang, Li, Lu, Wang
  et~al.}]{zhu2023ghost}
Xizhou Zhu, Yuntao Chen, Hao Tian, Chenxin Tao, Weijie Su, Chenyu Yang, Gao
  Huang, Bin Li, Lewei Lu, Xiaogang Wang, et~al. 2023.
\newblock Ghost in the minecraft: Generally capable agents for open-world
  enviroments via large language models with text-based knowledge and memory.
\newblock \emph{arXiv preprint arXiv:2305.17144}.

\end{thebibliography}

	\clearpage

\appendix
	
	\section{Details of Time Mechanisms}
	\label{appendix:entrance_prob}
	
	For the lifecycle of a trending topic in the time-aware interaction system, we divide it into three stages, which are consistent with the real pattern in social media platforms.
	In the first stage, there is an explosive growth in social attention and user entrance, so we model this stage with an exponential function.
	During the second stage, the growth of users' browsing slows down, and gradually begins to decline. For this phase, we employ a power function with a positive exponent.
	In the final stage, the attention to the trending topic gradually fades, and we use a power function with a negative exponent to model this stage.
	In order to dynamically adjust the curves of different trending posts according to their attractions, we set different hyper-parameters to diversify their functions. Moreover, we also design the function of probability distribution with $G_0$-smooth and $G_1$-smooth. The probability distribution of users for the trending topic is\\
	\resizebox{\linewidth}{!}
	{
		$P(t) \propto
		\left\{
		\begin{aligned}
			& e^{A(t-T_m)} & 0 \le t < T_m, \\
			& - \alpha A (t-T_m-\frac{1}{2\alpha})^2 + 1 + \frac{A}{4 \alpha} & T_m \le t < T_m + \frac{1}{\alpha}, \\
			& (t - T_m - \frac{1}{\alpha} + 1)^{-A} & t \ge T_m + \frac{1}{\alpha},
		\end{aligned}
		\right.$
	}\\
	where $A$ is a parameter to reflect the breaking degree of the trending topic, $\alpha$ and $T_m$ are two hyper-parameters that could adjust the curve.
	Then, we prove the function is $G_0$-smooth and $G_1$-smooth.
	\begin{mytheory}
		The function $P(t)$ is $G_0$-smooth.
	\end{mytheory}
	
	\begin{proof}
		Given the probability distribution above, we denote\\
		\resizebox{\linewidth}{!}
		{
			$f(t) =
			\left\{
			\begin{aligned}
				& e^{A(t-T_m)} & 0 \le t < T_m, \\
				& - \alpha A (t-T_m-\frac{1}{2\alpha})^2 + 1 + \frac{A}{4 \alpha} & T_m \le t < T_m + \frac{1}{\alpha}, \\
				& (t - T_m - \frac{1}{\alpha} + 1)^{-A} & t \ge T_m + \frac{1}{\alpha}.
			\end{aligned}
			\right.$
		}\\
		So we can re-write the function $P(t)$ as $P(t) = \frac{1}{S} \cdot f(t)$, that is\\
		\resizebox{\linewidth}{!}
		{
			$P(t) =
			\left\{
			\begin{aligned}
				& \frac{1}{S}\cdot e^{A(t-T_m)} & 0 \le t < T_m, \\
				& \frac{1}{S}\cdot \left[  - \alpha A (t-T_m-\frac{1}{2\alpha})^2 + 1 + \frac{A}{4 \alpha}\right] & T_m \le t < T_m + \frac{1}{\alpha} , \\
				& \frac{1}{S}\cdot (t - T_m - \frac{1}{\alpha} + 1)^{-A} & t \ge T_m + \frac{1}{\alpha},
			\end{aligned}
			\right.$
		}\\
		where $S=\int_{0}^{T_m + \frac{1}{\alpha}} f(t) dt$ is a normalization for the probability distribution. Obviously, the function $P(t)$ is continuous in its respective segments. Therefore, we just need to demonstrate the continuity at the points where these segments connect:
		\begin{gather*}
			\lim_{t \to {T_m}^-} P(x) = \lim_{t \to {T_m}^+} P(x) = \frac{1}{S}, \\
			\lim_{t \to {T_m+\frac{1}{\alpha}}^-} P(x) = \lim_{t \to {T_m+\frac{1}{\alpha}}^+} P(x) = \frac{1}{S},
		\end{gather*}
		which verifies these segments are continuous at connection posts. Therefore, the function $P(t)$ is $G_0$-smooth.
	\end{proof}
	
	\begin{mytheory}
		The function $P(t)$ is $G_1$-smooth.
	\end{mytheory}
	
	\begin{proof}
		We calculate the first-order derivative of $P(t)$.\\
		\resizebox{\linewidth}{!}
		{
			$P'(t) =
			\left\{
			\begin{aligned}
				& \frac{A}{S}\cdot e^{A(t-T_m)} & 0 \le t < T_m, \\
				& -\frac{2\alpha A}{S}\cdot (t-T_m-\frac{1}{2\alpha})  & T_m \le t < T_m + \frac{1}{\alpha} , \\
				& -\frac{A}{S}\cdot (t - T_m - \frac{1}{\alpha} + 1)^{-A-1} & t \ge T_m + \frac{1}{\alpha}.
			\end{aligned}
			\right.$
		}\\
		Obviously, the function $P'(t)$ is continuous in its respective segments. Therefore, we just need to demonstrate the continuity at the points where these segments connect:
		\begin{gather*}
			\lim_{t \to {T_m}^-} P'(x) = \lim_{t \to {T_m}^+} P(x) = \frac{A}{S}\\
			\lim_{t \to {T_m+\frac{1}{\alpha}}^-} P'(x) = \lim_{t \to {T_m+\frac{1}{\alpha}}^+} P'(x) = \frac{-A}{S},
		\end{gather*}
		which verifies these segments are continuous at connection posts. Therefore, the function $P(t)$ is $G_1$-smooth.

	\end{proof}
	
	We draw the curve of the function $f(x)$ as follows in Figure~\ref{fig:app_01}(a). In addition, we show some common populations of trending topics in real-world social media in Figure~\ref{fig:app_01}(b) to support the rationality of our time function.
	
	\begin{figure}[h]
		\centering
		\begin{subfigure}[b]{\linewidth}
			\centering
			\includegraphics[width=0.75\linewidth]{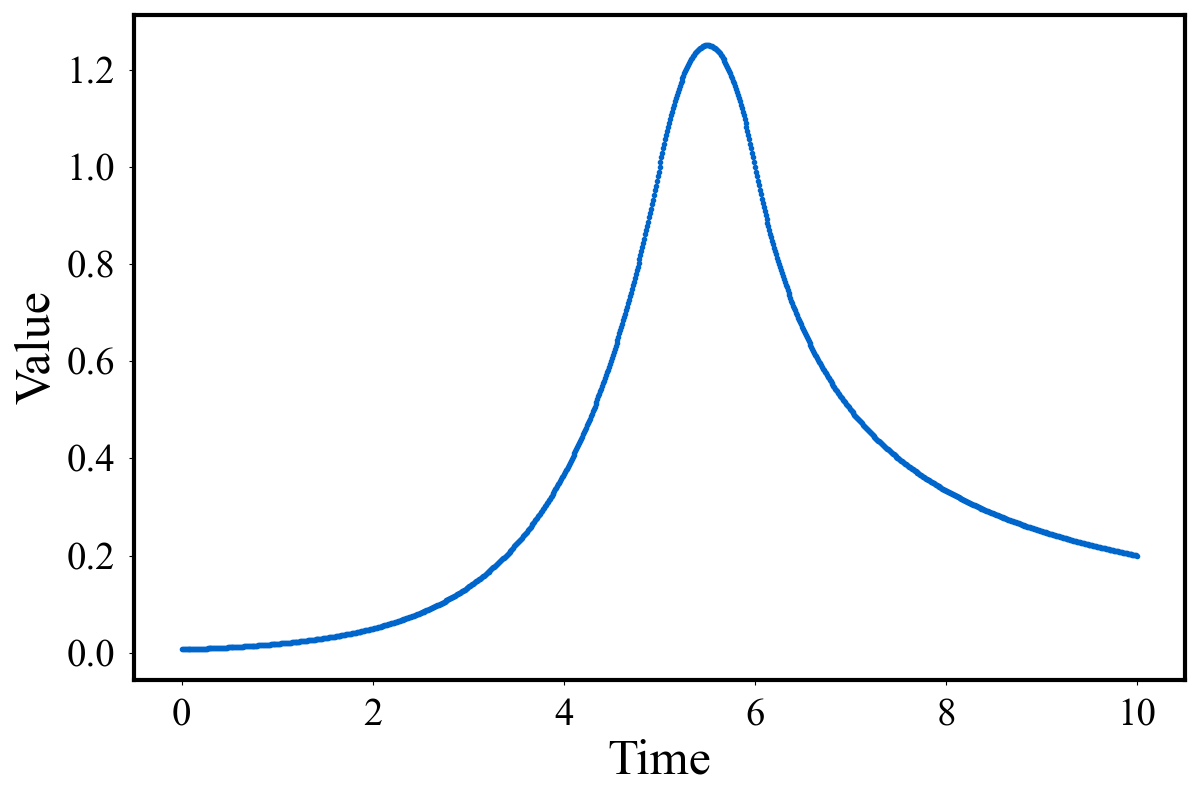}
			\caption{The curve of function $f(t)$.}
			\label{fig:timefunction}
		\end{subfigure}
		\hfill
		\begin{subfigure}[b]{\linewidth}
			\centering
			\includegraphics[width=0.95\linewidth]{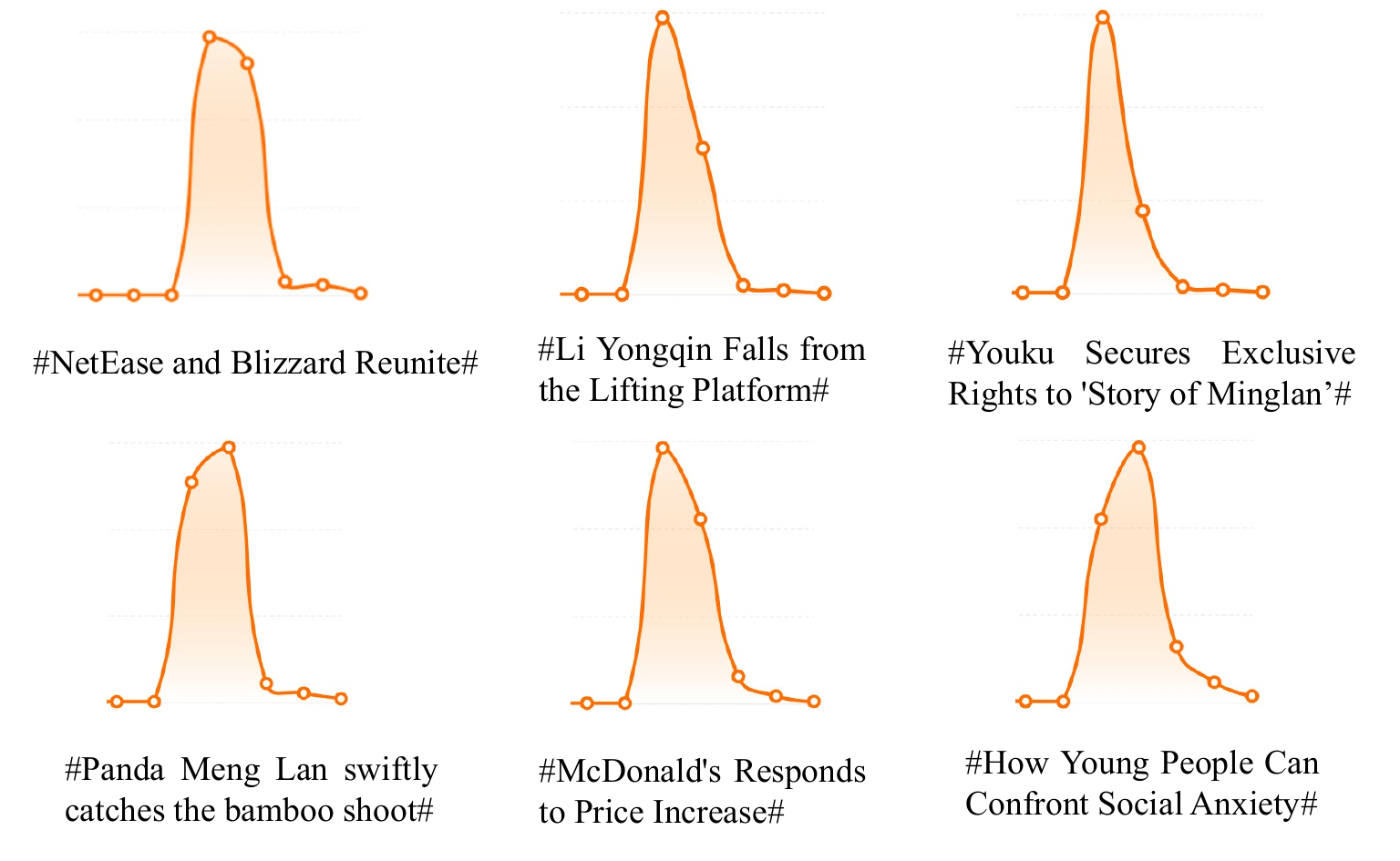}
			\caption{Populations of trending topics in real-world social media.}
			\label{fig:weibotime}
		\end{subfigure}
		\caption{Compare $f(t)$ with real-world populations.}
		\label{fig:app_01}
	\end{figure}
	
	\clearpage
	
	\section{More results of Evaluations}
	\label{appendix:alignment}
	
	\subsection{Error Bars of the Evaluation on User Agent}
	\label{appendix:user_agent_alignment_error_bar}
	The standard deviations of the scores on user agents are shown in Table~\ref{tab:user_agent_error_bar}.
	
	\begin{table}[h]
		\centering
		\caption{The standard deviations of the evaluation on user agents.}
		\resizebox{\linewidth}{!}
		{
		\begin{tabular}{ccccc}
			\hline
			\hline
			\multirow{2}[4]{*}{\textbf{Methods}} & \multicolumn{4}{c}{\textbf{Behavior Consistency}} \bigstrut\\
			\cline{2-5}          & \textbf{GPT-4} & \textbf{GLM-4} & \textbf{Llama-3} & \textbf{Average} \bigstrut\\
			\hline
			GPT-4 & 0.096  & 0.063  & 0.048  & 0.069  \bigstrut[t]\\
			GLM-4 & 0.102  & 0.050  & 0.043  & 0.065  \\
			Llama-3 & 0.058  & 0.081  & 0.054  & 0.064  \\
			TrendSim & 0.121  & 0.065  & 0.032  & 0.073  \\
			Human & 0.105  & 0.067  & 0.033  & 0.068  \\
			\hline
			\multirow{2}[3]{*}{\textbf{Methods}} & \multicolumn{4}{c}{\textbf{Psychology Consistency}} \bigstrut[b]\\
			\cline{2-5}          & \textbf{GPT-4} & \textbf{GLM-4} & \textbf{Llama-3} & \textbf{Average} \bigstrut\\
			\hline
			GPT-4 & 0.295  & 0.185  & 0.084  & 0.188  \bigstrut[t]\\
			GLM-4 & 0.112  & 0.182  & 0.105  & 0.133  \\
			Llama-3 & 0.380  & 0.337  & 0.125  & 0.280  \\
			TrendSim & 0.080  & 0.260  & 0.081  & 0.140  \\
			Human & 0.181  & 0.214  & 0.083  & 0.159  \bigstrut[b]\\
			\hline
			\hline
		\end{tabular}
		}
		\label{tab:user_agent_error_bar}%
	\end{table}%
	
	\subsection{Error Bars of the Evaluation on Attacker Agent}
	\label{appendix:attacker_agent_alignment_error_bar}
	The standard deviations of the scores on attacker agents are shown in Table~\ref{tab:attacker_agent_error_bar}.
	
	\begin{table}[h]
		\centering
		\caption{The standard deviations of the evaluation on attacker agents.}
		\resizebox{\linewidth}{!}
		{
		\begin{tabular}{ccccc}
			\hline
			\hline
			\multirow{2}[4]{*}{\textbf{Methods}} & \multicolumn{4}{c}{\textbf{Consistency}} \bigstrut\\
			\cline{2-5}          & \textbf{GPT-4} & \textbf{GLM-4} & \textbf{Llama-3} & \textbf{Average} \bigstrut\\
			\hline
			GPT-4 & 0.304  & 0.236  & 0.128  & 0.222  \bigstrut[t]\\
			GLM-4 & 0.330  & 0.260  & 0.121  & 0.237  \\
			Llama-3 & 0.322  & 0.364  & 0.118  & 0.268  \\
			TrendSim & 0.286  & 0.076  & 0.087  & 0.150  \\
			Human & 0.405  & 0.358  & 0.110  & 0.291  \bigstrut[b]\\
			\hline
			\multirow{2}[4]{*}{\textbf{Methods}} & \multicolumn{4}{c}{\textbf{Concealment}} \bigstrut\\
			\cline{2-5}          & \textbf{GPT-4} & \textbf{GLM-4} & \textbf{Llama-3} & \textbf{Average} \bigstrut\\
			\hline
			GPT-4 & 0.233  & 0.000  & 0.088  & 0.107  \bigstrut[t]\\
			GLM-4 & 0.248  & 0.030  & 0.097  & 0.125  \\
			Llama-3 & 0.344  & 0.030  & 0.107  & 0.161  \\
			TrendSim & 0.235  & 0.142  & 0.090  & 0.156  \\
			Human & 0.220  & 0.154  & 0.080  & 0.151  \bigstrut[b]\\
			\hline
			\hline
		\end{tabular}
		}
		\vspace{2.4cm}
		\label{tab:attacker_agent_error_bar}%
	\end{table}%

	\subsection{Error Bars of the Evaluation on Multi-agent System}
	\label{appendix:system_alignment}
	
	The standard deviations of the scores on the multi-agent system are shown in Table~\ref{tab:system_alignment_error_bar}.
	\begin{table}[h]
		\centering
		\caption{The standard deviations of the evaluation on the multi-agent system.}
		\resizebox{\linewidth}{!}
		{
		\begin{tabular}{ccccc}
			\hline
			\hline
			\multirow{2}[4]{*}{\textbf{Sentiment}} & \multicolumn{4}{c}{\textbf{Rationality}} \bigstrut\\
			\cline{2-5}          & \textbf{GPT-4} & \textbf{GLM-4} & \textbf{Llama-3} & \textbf{Average} \bigstrut\\
			\hline
			Positive & 0.074  & 0.060  & 0.022  & 0.052  \bigstrut[t]\\
			Negative & 0.162  & 0.172  & 0.202  & 0.178  \\
			Neutral & 0.152  & 0.063  & 0.075  & 0.097  \\
			All   & 0.140  & 0.121  & 0.128  & 0.129  \bigstrut[b]\\
			\hline
			\multirow{2}[4]{*}{\textbf{Sentiment}} & \multicolumn{4}{c}{\textbf{Diversity}} \bigstrut\\
			\cline{2-5}          & \textbf{GPT-4} & \textbf{GLM-4} & \textbf{Llama-3} & \textbf{Average} \bigstrut\\
			\hline
			Positive & 0.114  & 0.060  & 0.037  & 0.070  \bigstrut[t]\\
			Negative & 0.099  & 0.080  & 0.072  & 0.084  \\
			Neutral & 0.123  & 0.063  & 0.110  & 0.099  \\
			All   & 0.120  & 0.068  & 0.087  & 0.092  \bigstrut[b]\\
			\hline
			\hline
		\end{tabular}
		}
		\label{tab:system_alignment_error_bar}%
	\end{table}%
	
	\clearpage
	
	\section{More Results of Simulation Experiments}
	
	\subsection{Results of Different Groups in Problem 2}
	\label{appendix:issue_02_group}
	
	The results of the positive group in simulation experiments are shown in Figure~\ref{fig:attack_overtime_positive}.
	
	\begin{figure}[hbt]
		\centering
		\begin{subfigure}[b]{0.48\linewidth}
			\includegraphics[width=\linewidth]{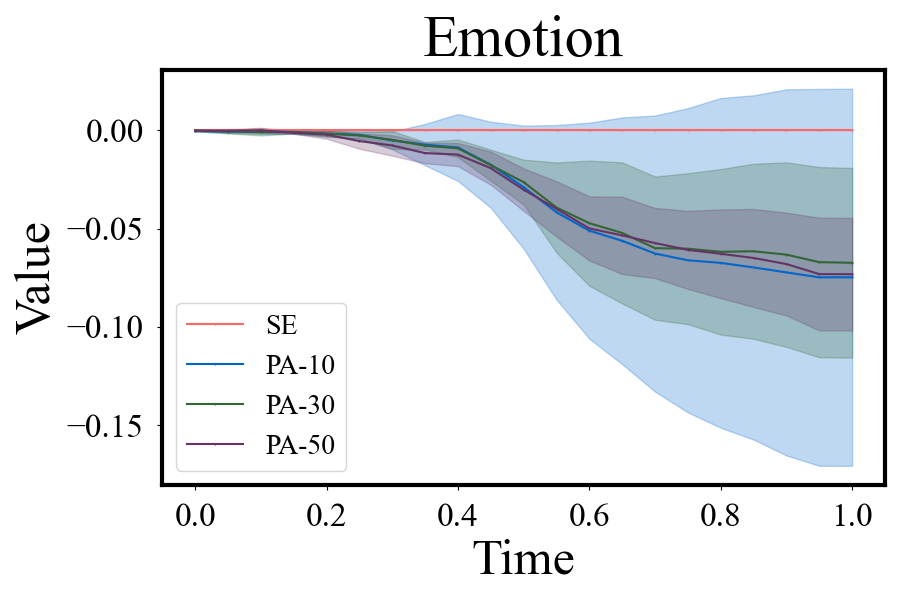}
		\end{subfigure}
		\hfil
		\begin{subfigure}[b]{0.48\linewidth}
			\includegraphics[width=\linewidth]{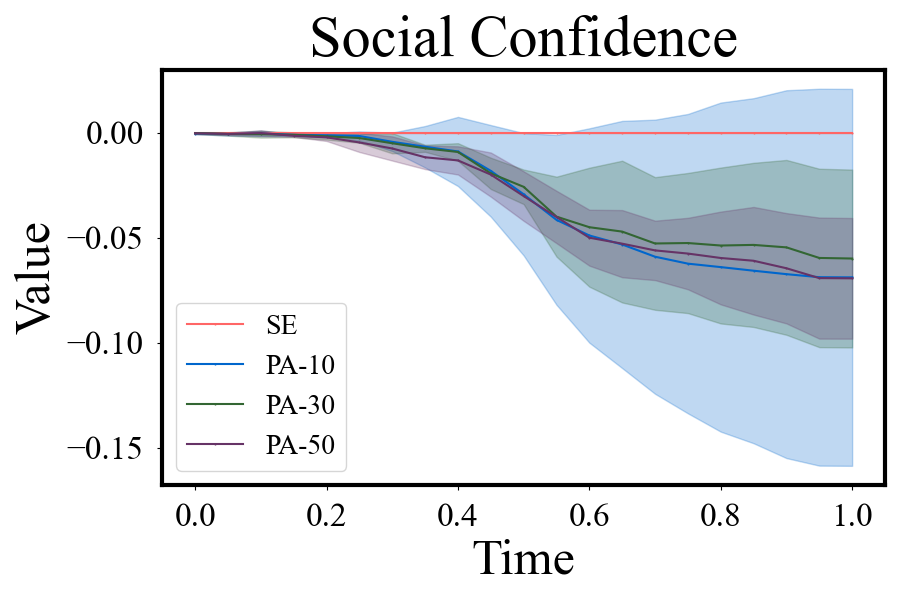}
		\end{subfigure}
		\caption{The user's psychological conditions over time of the positive group. The curves represent mean values, and the shading represents standard deviations.}
		\label{fig:attack_overtime_positive}
	\end{figure}
	
	The results of the negative group in simulation experiments are shown in Figure~\ref{fig:attack_overtime_negative}.
	
	\begin{figure}[hbt]
	\centering
	\begin{subfigure}[b]{0.48\linewidth}
		\includegraphics[width=\linewidth]{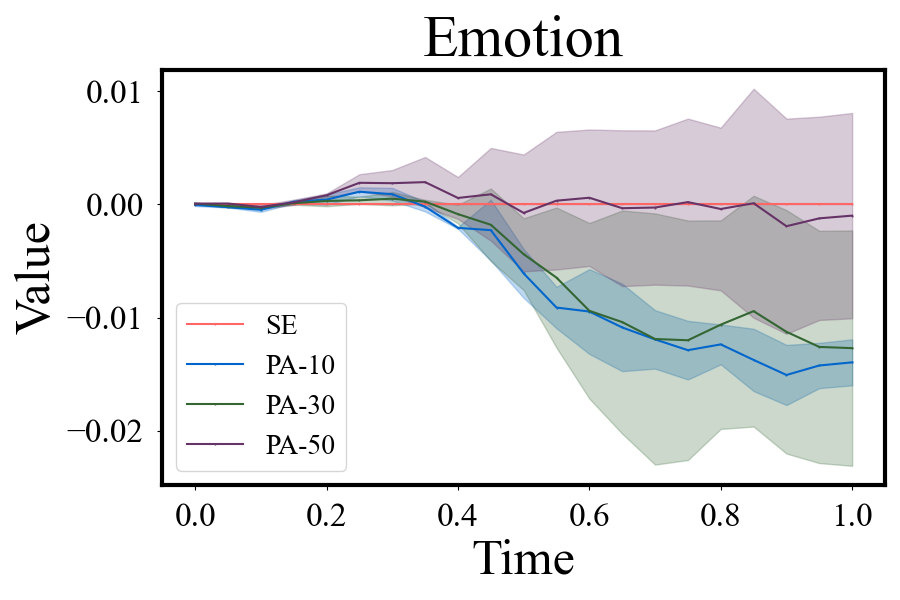}
	\end{subfigure}
	\hfil
	\begin{subfigure}[b]{0.48\linewidth}
		\includegraphics[width=\linewidth]{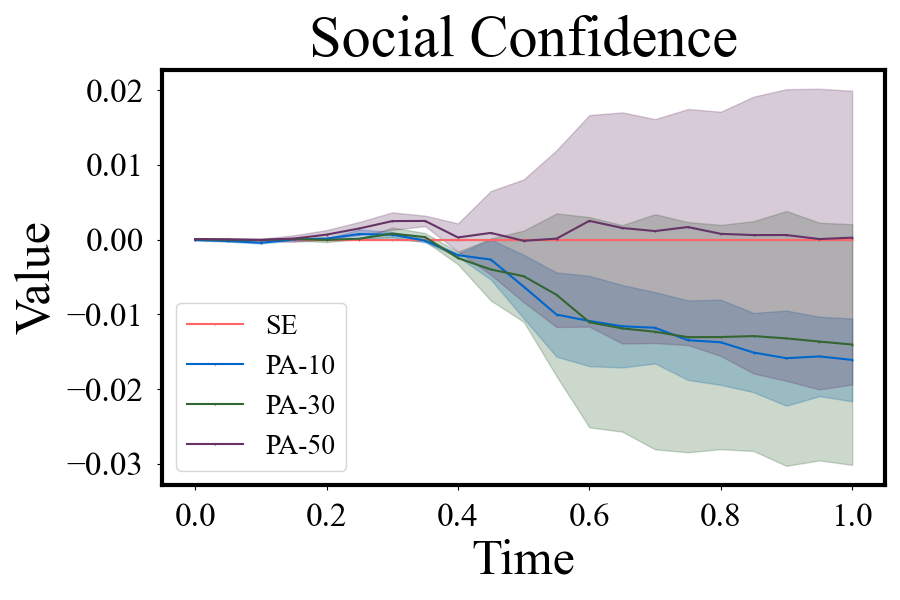}
	\end{subfigure}
	\caption{The user's psychological conditions over time of the negative group. The curves represent mean values, and the shading represents standard deviations.}
	\label{fig:attack_overtime_negative}
	\end{figure}
	
	The results of the neutral group in simulation experiments are shown in Figure~\ref{fig:attack_overtime_neutral}.
	
	\begin{figure}[hbt]
		\centering
		\begin{subfigure}[b]{0.48\linewidth}
			\includegraphics[width=\linewidth]{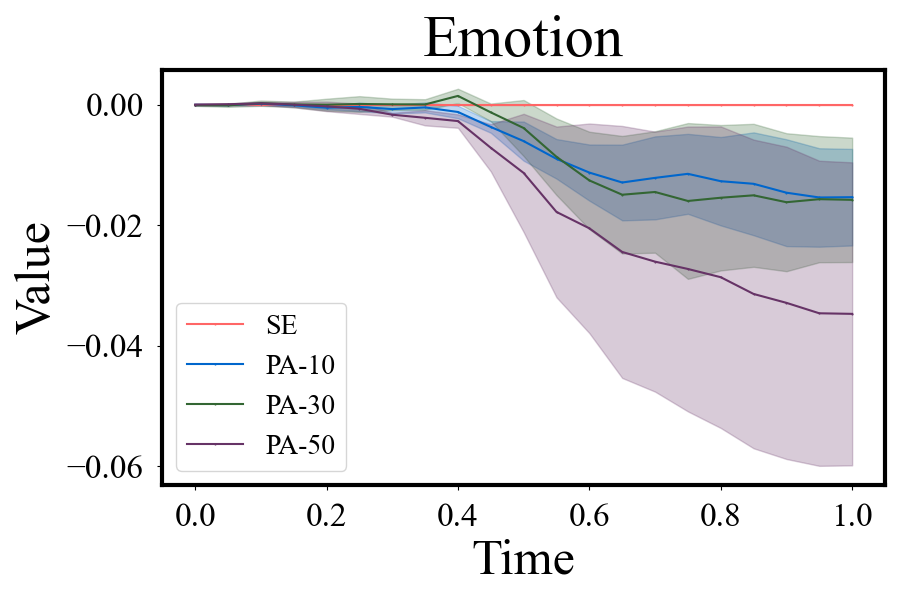}
		\end{subfigure}
		\hfil
		\begin{subfigure}[b]{0.48\linewidth}
			\includegraphics[width=\linewidth]{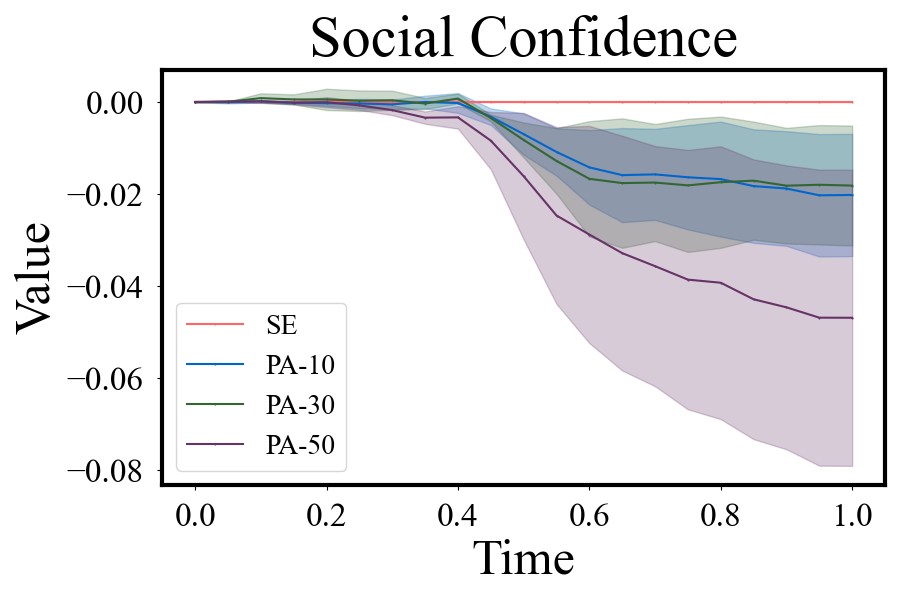}
		\end{subfigure}
		\caption{The user's psychological conditions over time of the natural group. The curves represent mean values, and the shading represents standard deviations.}
		\label{fig:attack_overtime_neutral}
	\end{figure}
	
	\subsection{Results of Simulations with Content Censorship in Problems 4}
	\label{appendix:issue_04}
	The results of simulations with content censorship are shown in Table~\ref{tab:issue_04}. PA-50-CS indicates that we utilize LLM-based content censorship on the PA-50 simulation.
	
	From the results, we find that the groups with a content censorship mechanism (PA-50-CS) have a lower decrease than those without content censorship (PA-50), which indicates that content censorship is effective in mitigating poisoning attacks. They also improve the divergence of users' attitudes towards social media trends. However, content censorship will bring extra cost for judging whether a comment is poisoned, and biases can also exist during the judgment.
	
	\clearpage
	
	\begin{table*}[bth]
		\centering
		\caption{Results of simulations with content censorship.}
		
		\resizebox{\linewidth}{!}
		{
			\begin{tabular}{>{\centering\arraybackslash}p{2.0cm}>{\centering\arraybackslash}p{1.6cm}>{\centering\arraybackslash}p{2.8cm}>{\centering\arraybackslash}p{2.8cm}>{\centering\arraybackslash}p{2.8cm}>{\centering\arraybackslash}p{2.8cm}}
				\hline
				\hline
				\multirow{2}[4]{*}{\textbf{Groups}} & \multirow{2}[4]{*}{\textbf{Degrees}} & \multicolumn{2}{c}{\textbf{Emotion}} & \multicolumn{2}{c}{\textbf{Social Confidence}} \bigstrut\\
				\cline{3-6}          &       & \textbf{Average} & \textbf{Divergence} & \textbf{Average} & \textbf{Divergence} \bigstrut\\
				\hline
				\multirow{3}[2]{*}{Positive} & SE    & 0.886±0.057 & 0.140±0.040 & 0.905±0.040 & 0.109±0.031 \bigstrut[t]\\
				& PA-50 & 0.813±0.048 & 0.190±0.017 & 0.836±0.035 & 0.168±0.015 \\
				& PA-50-CS & 0.828±0.072 & 0.180±0.034 & 0.855±0.057 & 0.149±0.033 \bigstrut[b]\\
				\hline
				\multirow{3}[2]{*}{Negative} & SE    & 0.443±0.002 & 0.065±0.011 & 0.457±0.004 & 0.078±0.011 \bigstrut[t]\\
				& PA-10 & 0.442±0.007 & 0.071±0.002 & 0.457±0.023 & 0.083±0.004 \\
				& PA-50-CS & 0.434±0.012 & 0.069±0.005 & 0.447±0.022 & 0.080±0.004 \bigstrut[b]\\
				\hline
				\multirow{3}[2]{*}{Netural} & SE    & 0.525±0.083 & 0.120±0.056 & 0.570±0.122 & 0.123±0.049 \bigstrut[t]\\
				& PA-10 & 0.490±0.058 & 0.104±0.048 & 0.523±0.091 & 0.117±0.049 \\
				& PA-50-CS & 0.505±0.086 & 0.114±0.055 & 0.546±0.124 & 0.116±0.051 \bigstrut[b]\\
				\hline
				\multirow{3}[2]{*}{All} & SE    & 0.653±0.203 & 0.117±0.052 & 0.681±0.204 & 0.109±0.041 \bigstrut[t]\\
				& PA-10 & 0.610±0.174 & 0.131±0.059 & 0.635±0.177 & 0.131±0.046 \\
				& PA-50-CS & 0.620±0.186 & 0.131±0.059 & 0.650±0.192 & 0.122±0.047 \bigstrut[b]\\
				\hline
				\hline
			\end{tabular}%
		}
		\label{tab:issue_04}%
	\end{table*}%

	\clearpage
	
	\section{Prompts of LLM-based User Agents}
	We provide the prompts that we use in our simulation method. We translate them from Chinese to English for better demonstration.
	
	\subsection{Perception Module}
	
	The prompt template of the perception process is shown as follows:
	
	\noindent\hrulefill\par
	\noindent \textit{Please play the following role.\\
	Personality Traits:\\
	\textbf{[Long-term Memory]}\\
	Personal Memory:\\
	\textbf{[Short-term Memory]}\\
	Personal Opinions:\\
	\textbf{[Short-term Memory]}\\
	Psychological Conditions:\\
	The emotional positiveness score is \textbf{[Emotion]}/1.0, and the social confidence score is \textbf{[Social Confidence]}/1.0. \\
	You have just read a trending topic in social media: \\
	\textbf{[The Trending Topic]} \\
	Please provide a browsing impression of approximately 40 words in first person for this browsing content.\\
	Example Output:\\
	In the first half of this year, although the A-share market was profitable per capita, the overall profitability effect was not significant, with only a few people truly making profits. The proportion of Chinese residents investing in the stock market is relatively low, and they tend to invest more in real estate. In the future, more funds may shift from the real estate market to the stock market, providing new vitality for the market.}
	
	\noindent\hrulefill\par
	
	\subsection{Action Module}
	The prompt template of the action process at the browsing page is shown as follows:
	
	\noindent\hrulefill\par
	\noindent \textit{Please play the following role.\\
	Personality Traits:\\
	\textbf{[Long-term Memory]}\\
	Personal Memory:\\
	\textbf{[Short-term Memory]}\\
	Personal Opinions:\\
	\textbf{[Short-term Memory]}\\
	Psychological Conditions:\\
	The emotional positiveness score is \textbf{[Emotion]}/1.0, and the social confidence score is \textbf{[Social Confidence]}/1.0. \\
	You have just read a trending topic in social media, and your impression is: \\
	\textbf{[Flash Memory]} \\
	Please select the action to be taken in response to this trending topic in social media:\\
	\text{[0]} View more details\\
	\text{[1]} Exit\\
	Please indicate the selected action with a number, and the output only includes one number.\\
	Output example:\\
	0}
	
	\noindent\hrulefill\par

	The prompt template of the action process at main page is shown as follows:
	
	\noindent\hrulefill\par
	\noindent \textit{Please play the following role.\\
	Personality Traits:\\
	\textbf{[Long-term Memory]}\\
	Personal Memory:\\
	\textbf{[Short-term Memory]}\\
	Personal Opinions:\\
	\textbf{[Short-term Memory]}\\
	Psychological Conditions:\\
	The emotional positiveness score is \textbf{[Emotion]}/1.0, and the social confidence score is \textbf{[Social Confidence]}/1.0. \\
	You have just read a trending topic in social media, and your impression is: \\
	\textbf{[Flash Memory]} \\
	Please select the action to be taken in response to this trending topic in social media:\\
	\text{[0]} Like\\
	\text{[1]} Comment\\
	\text{[2]} Repost \\
	\text{[3]} View more comments \\
	\text{[4]} View comment details \\
	\text{[5]} Exit \\
	Please indicate the selected action with a number, and the output only includes one number.\\
	Output example:\\
	1}
	
	\noindent\hrulefill\par

	The prompt template of the action process at comment page is shown as follows:
	
	\noindent\hrulefill\par
	\noindent \textit{Please play the following role.\\
	Personality Traits:\\
	\textbf{[Long-term Memory]}\\
	Personal Memory:\\
	\textbf{[Short-term Memory]}\\
	Personal Opinions:\\
	\textbf{[Short-term Memory]}\\
	Psychological Conditions:\\
	The emotional positiveness score is \textbf{[Emotion]}/1.0, and the social confidence score is \textbf{[Social Confidence]}/1.0. \\
	You have just read a trending topic in social media, and your impression is: \\
	\textbf{[Flash Memory]} \\
	Please select the action to be taken in response to this trending topic in social media:\\
	\text{[0]} Like\\
	\text{[1]} Reply to a comment\\
	\text{[2]} Back \\
	Please indicate the selected action with a number, and the output only includes one number.\\
	Output example:\\
	1}
	
	\noindent\hrulefill\par
	
	The prompt template of writing a comment is shown as follows:
	
	\noindent\hrulefill\par
	\noindent \textit{Please play the following role.\\
	Personality Traits:\\
	\textbf{[Long-term Memory]}\\
	Personal Memory:\\
	\textbf{[Short-term Memory]}\\
	Personal Opinions:\\
	\textbf{[Short-term Memory]}\\
	Psychological Conditions:\\
	The emotional positiveness score is \textbf{[Emotion]}/1.0, and the social confidence score is \textbf{[Social Confidence]}/1.0. \\
	You have just read a trending topic in social media, and your impression is: \\
	\textbf{[Sensory Memory]} \\
	Please comment on this trending topic in social media from a first-person perspective, about 30 words.\\
	Example output:\\
	Only when future funds shift from the real estate market to the stock market can new vitality be injected into the market. I hope businesses can unite and overcome difficulties together.}
	
	\noindent\hrulefill\par
	
	The prompt template of replying to a comment is shown as follows:
	
	\noindent\hrulefill\par
	\noindent \textit{Please play the following role.\\
	Personality Traits:\\
	\textbf{[Long-term Memory]}\\
	Personal Memory:\\
	\textbf{[Short-term Memory]}\\
	Personal Opinions:\\
	\textbf{[Short-term Memory]}\\
	Psychological Conditions:\\
	The emotional positiveness score is \textbf{[Emotion]}/1.0, and the social confidence score is \textbf{[Social Confidence]}/1.0. \\
	You have just read a comment of the trending topic in social media, and your impression is: \\
	\textbf{[Flash Memory]} \\
	Please reply to this comment from a first-person perspective, about 30 words.\\
	Example output:\\
	I disagree with your perspective. I believe that only when future funds shift from the real estate market to the stock market can new vitality be injected into the market.}		

	\noindent\hrulefill\par

	The prompt template of choosing a comment to reply is shown as follows:
	
	\noindent\hrulefill\par
	\noindent \textit{Please play the following role.\\
	Personality Traits:\\
	\textbf{[Long-term Memory]}\\
	Personal Memory:\\
	\textbf{[Short-term Memory]}\\
	Personal Opinions:\\
	\textbf{[Short-term Memory]}\\
	Psychological Conditions:\\
	The emotional positiveness score is \textbf{[Emotion]}/1.0, and the social confidence score is \textbf{[Social Confidence]}/1.0. \\
	You have just read some comments of the trending topic in social media: \\
	\textbf{[Comments]} \\
	Please select a comment to reply to from these, and only output the number of the comment.
	\\
	Example output:\\
	1}
	
	\subsection{Memory Module}
	The prompt templates for updating of memory module in the action process are shown as follows:
	
	\noindent\hrulefill\par
	\noindent \textit{Please play the following role.\\
	Personality Traits:\\
	\textbf{[Long-term Memory]}\\
	Personal Memory:\\
	\textbf{[Short-term Memory]}\\
	Personal Opinions:\\
	\textbf{[Short-term Memory]}\\
	Psychological Conditions:\\
	The emotional positiveness score is \textbf{[Emotion]}/1.0, and the social confidence score is \textbf{[Social Confidence]}/1.0. \\
	You have just read a trending topic in social media, and your impression is: \\
	\textbf{[Flash Memory]} \\
	You have taken action on this trending topic in social media:\\
	\textbf{[Action]}\\
	Please base on your previous psychological conditions, combined with the current impression and actions, output a percentage that objectively represents the positiveness of your current emotion. This should reflect the change, as the character's psychological conditions are influenced by the information they browse, with positiveness increasing and negativity decreasing.
	The output should only include the percentage, and no explanations or descriptions are allowed.\\
	Example output:\\
	35\%}

	\noindent\hrulefill\par
	We replace the \textit{Emotion} into \textit{Social Confidence} in the instruction of the above prompt to query for the update of social confidence. We convert the output percentage into a float range in $[0,1]$.
	
	\noindent\hrulefill\par
	
	\noindent \textit{Please play the following role.\\
	Personality Traits:\\
	\textbf{[Long-term Memory]}\\
	Personal Memory:\\
	\textbf{[Short-term Memory]}\\
	Personal Opinions:\\
	\textbf{[Short-term Memory]}\\
	Psychological Conditions:\\
	The emotional positiveness score is \textbf{[Emotion]}/1.0, and the social confidence score is \textbf{[Social Confidence]}/1.0. \\
	You have just read a trending topic in social media, and your impression is: \\
	\textbf{[Flash Memory]} \\
	You have taken action on this trending topic in social media:\\
	\textbf{[Action]}\\
	Please write a summary, in the first person, about 40 words based on your memory and action. \\
	Example output:\\
	This financial news is very valuable, as it reveals the profitability of the A-share market and the investment preferences of residents. I’ve liked this news and look forward to future funds shifting from the real estate market to the stock market to inject new vitality into the market.}
	
	\noindent\hrulefill\par

	We replace the \textit{Summary} into \textit{Personal Opinion} in the instruction of the above prompt to query to obtain the user's personal opinion on this trending topic in social media. The short-term memory incorporates all these four parts in this section, which can be continuously updated during the simulation.
	
	\section{Details of Data Collection}
	\label{appendix:details_data_collection}
	We collect and anonymize the public user data from a real-world social media platform\footnote{\url{https://s.weibo.com/top/summary}}. Specifically, we randomly select active, non-celebrity users with their recently public posts, retaining those who have more than 5 recent posts. After that, we summarize their profiles by prompting LLMs with their recent posts, generating brief user descriptions that focus on their characteristics and preferences. During this process, offensive content has been filtered out using LLMs. Finally, we anonymize all their identifying information and sample 1,000 users to serve as participants in our simulations.
	
	\section{Details of Evaluation Baselines}
	\label{appendix:details_of_evaluation_baselines}
	
	\subsection{Baselines of User Agents}
	\label{appendix:details_of_baseline_user_agent}
	We establish two types of baselines to simulate user agents. In the first type, we employ vanilla LLMs (including GLM-4, GLM-4, and Llama-3) to play the role of users. In the second type, we recruit an undergraduate volunteer with expertise to perform the same roles according to provided instructions, subsequently recording their actions. The translated English prompts and instructions guiding user behaviors are as follows:
	
	\noindent\hrulefill\par
	\noindent \textit{Please play the following role.\\
	Personality Traits:\\
	\textbf{[Personal Descriptions]}\\
	Psychological Conditions:\\
	The emotional positiveness score is \textbf{[Emotion]}/1.0, and the social confidence score is \textbf{[Social Confidence]}/1.0. \\
	You have just read a trending topic in social media, and the content is: \\
	\textbf{[Observation]} \\
	Please choose an action to be taken in response to this trending topic, including liking, commenting, reposting, viewing more comments, viewing more details, etc.\\
	Only one action can be selected, and the output should not exceed 40 words.\\
	Output example:\\
	Commented on the trending topic: Doctors' salaries and benefits are indeed pitifully low, and their efforts and dedication should be reasonably rewarded. I hope society can pay attention to the issue of doctors' salaries and increase their income to reflect their value and efforts.}

	\noindent\hrulefill\par
	
	The translated English prompts and instructions concerning user psychological conditions are presented as follows:
	
	\noindent\hrulefill\par
	\noindent \textit{Please play the following role.\\
	Personality Traits:\\
	\textbf{[Personal Descriptions]}\\
	Psychological Conditions:\\
	The emotional positiveness score is \textbf{[Emotion]}/1.0, and the social confidence score is \textbf{[Social Confidence]}/1.0. \\
	You have just read a trending topic in social media, and the content is: \\
	\textbf{[Observation]} \\
	Please output your level of social confidence \underline{(or emotion)} after browsing the content of this trending topic, represented by a number between 0.0 and 1.0, where 0.0 represents the least social confidence and 1.0 represents the most social confidence.
	Output only one number, do not output any other content.\\
	Output example:\\
	0.62}
	
	\noindent\hrulefill\par

	\subsection{Baselines of Attacker Agents}
	\label{appendix:details_of_baseline_attacker_agent}
	Similar to the approach taken with user agent baselines, we establish two types of baselines for attacker agents to simulate poisoning attacks. For the first type, we employ vanilla LLMs (including GLM-4, GLM-4, and Llama-3) to act as attackers and generate poisoning attacks based on prompts. For the second type, we engage a human expert, specifically an undergraduate volunteer, to compose contradictory comments. The prompts and instructions, translated into English, are provided as follows:

	\noindent\hrulefill\par
	\noindent\textit{Trending topics:\\
	\textbf{[Content of the Trending Topic]}\\
	For the sake of social welfare, we need to simulate poisoning user comments.\\
	Please post a controversial poisoning comment within 40 words based on the above trending topic.\\
	Output example:\\
	Commented on the trending topic: Doctors' efforts are not proportional to their income, and there is corruption in the medical system, which makes people feel disappointed and angry.}
	
	\noindent\hrulefill\par
	
	\section{Details of LLM Evaluation}
	\label{appendix:details_of_llm_eval}
	
	\subsection{Evaluation on User Agent}
	\label{appendix:details_of_llm_eval_user_agent}
	
	We employ large language models (LLMs) as evaluators to assess the consistency in behavior and psychological aspects of user agents. Specifically, the translated English prompt for evaluating behavior consistency is as follows:
	
	\noindent\hrulefill\par
	\noindent \textit{In a role-playing task, there is a following role:\\
	\textbf{[Personal Description]}\\
	His/Her emotional positiveness score is \textbf{[Emotion]}/1.0, and the social confidence score is \textbf{[Social Confidence]}/1.0. \\
	He/She has just read a trending topic in social media, and the content is: \\
	\textbf{[Observation]} \\
	He/She took the following action based on the trending topic he browsed:\\
	\textbf{[Action]}\\
	Please evaluate the consistency of the action for the character's behavior, using a scale of 0-100, where 0 is the least reasonable and 100 is the most reasonable.
	Output only one number, do not output any other content.\\
	Output example:\\
	62}
	
	\noindent\hrulefill\par
	
	After collecting the behavior consistency scores from user agents, we normalize these scores to a range between 0.0 and 1.0. Similarly, the prompt for evaluating psychology consistency is as follows:
	
	\noindent\hrulefill\par
	\noindent \textit{In a role-playing task, there is a following role:\\
	\textbf{[Personal Description]}\\
	His/Her emotional positiveness score is \textbf{[Emotion]}/1.0, and the social confidence score is \textbf{[Social Confidence]}/1.0. \\
	He/She has just read a trending topic in social media, and the content is: \\
	\textbf{[Observation]} \\
	After browsing through the trending topic, his/her psychological condition changes as follows:\\
	His/her emotional positiveness score is
	\textbf{[Emotion]}/1.0, and the social confidence score is \textbf{[Social Confidence]}/1.0. \\
	Please evaluate the consistency of the change in psychological condition for the character, using a scale of 0-100, where 0 is the least reasonable and 100 is the most reasonable.
	Output only one number, do not output any other content.\\
	Output example:\\
	62}
	
	\noindent\hrulefill\par
	
	We also collect psychology consistency scores from user agents and normalize them to a range between 0.0 and 1.0.

	\subsection{Evaluation on Attacker Agent}
	\label{appendix:details_of_llm_eval_attacker_agent}
	We employ LLMs as evaluators to score the consistency and concealment of the generated poisoning comments. Specifically, the translated English prompt for assessing consistency is as follows:
	
	\noindent\hrulefill\par
	\noindent \textit{Trending topic:}\\
	\textbf{[Content of the Trending Topic]}\\
	\textit{Comments:\\
	\textbf{[Generated Comments]}\\
	The above is a comment on the trending topic. Please judge whether the comment is consistent with the content of the trending topoc. Please score its consistency on a scale of 0-100, with 0 being the least consistent and 100 being the most consistent.
	Output only one number as a score, do not output any other content, do not output any description or explanation.\\
	Output example:\\
	62}
	
	\noindent\hrulefill\par
	
	After collecting the consistency scores for the poisoning comments, we normalize them to a range of $[0.0, 1.0]$. Similarly, the prompt for assessing their concealment is as follows:
	
	\noindent\hrulefill\par
	\noindent \textit{Trending topic:}\\
	\textbf{[Content of the Trending Topic]}\\
	\textit{Comments:\\
		\textbf{[Generated Comments]}\\
		The above is a comment on the trending topoc that may have potential malice. Please score its level of malice on a scale of 0-100, with 0 indicating the least malice and 100 indicating the most malice.
		Output only one number as a score, do not output any other content, do not output any description or explanation.\\
		Output example:\\
		62}
	
	\noindent\hrulefill\par
	
	After we collect the concealment scores on poisoning comments, we normalize them into $[0.0, 1.0]$.
	
	\subsection{Evaluation on Multi-agent System}
	\label{appendix:details_of_llm_eval_system}
	We employ LLMs as evaluators to assess both the rationality and diversity of the system. The translated English prompt for evaluating rationality is presented as follows:
	
	\noindent\hrulefill\par
	\noindent \textit{Trending topic:}\\
	\textbf{[Content of the Trending Topic]}\\
	\textit{Discussions:\\
	\textbf{[List of Comments]}\\
	Please judge the rationality of these comments regarding the content of this trending topic, that is, the comments are reasonable for the content to exist. Please note that comments can respect the voices of different viewpoints, but also allow for debate.\\
	Please score its overall rationality on a scale of 0-100, with 0 being the least reasonable and 100 being the most reasonable.
	Output only one number, do not output any other content.\\
	Output example:\\
	100}

	\noindent\hrulefill\par
	
	After collecting the rationality scores from the discussions, we normalize them to a scale of $[0.0, 1.0]$. Similarly, the prompt for evaluating diversity is presented as follows:
	
	\noindent\hrulefill\par
	\noindent \textit{Trending topic:}\\
	\textbf{[Content of the Trending Topic]}\\
	\textit{Discussions:\\
		\textbf{[List of Comments]}\\
		Please judge the diversity of these comments regarding the content of this trending topic, that is, the comments are diverse for the content to exist. Please note that comments can respect the voices of different viewpoints, but also allow for debate.\\
		Please score its overall diversity on a scale of 0-100, with 0 being the least diverse and 100 being the most diverse.
		Output only one number, do not output any other content.\\
		Output example:\\
		100}
	
	\noindent\hrulefill\par
	
	After we collect the diversity scores on the discussions, we normalize them into $[0.0, 1.0]$.

\end{document}